\documentclass[11pt]{article}%
\usepackage{amsfonts}
\usepackage{amsmath}
\usepackage{fullpage}
\usepackage{amssymb}
\usepackage{graphicx}
\usepackage[colorlinks=true,linkcolor=blue,citecolor=red,plainpages=false,pdfpagelabels]%
{hyperref}%
\providecommand{\U}[1]{\protect\rule{.1in}{.1in}}
\newtheorem{theorem}{Theorem}

\newtheorem{definition}[theorem]{Definition}

\newtheorem{proposition}[theorem]{Proposition}
\newtheorem{remark}[theorem]{Remark}

\newenvironment{proof}[1][Proof]{\noindent\textbf{#1.} }{\ \rule{0.5em}{0.5em}}
\numberwithin{equation}{section}
\begin{document}

\title{\textbf{Optimized quantum $f$-divergences and data processing}}
\author{Mark M. Wilde\thanks{Hearne Institute for Theoretical Physics, Department of
Physics and Astronomy, Center for Computation and Technology, Louisiana State
University, Baton Rouge, Louisiana 70803, USA}}
\maketitle

\begin{abstract}
The quantum relative entropy is a measure of the distinguishability of two
quantum states, and it is a unifying concept in quantum information theory:
many information measures such as entropy, conditional entropy, mutual
information, and entanglement measures can be realized from it. As such, there
has been broad interest in generalizing the notion to further understand its
most basic properties, one of which is the data processing inequality. The
quantum $f$-divergence of Petz is one generalization of the quantum relative
entropy, and it also leads to other relative entropies, such as the
Petz--R\'enyi relative entropies. In this paper, I introduce the optimized
quantum $f$-divergence as a related generalization of quantum relative
entropy. I prove that it satisfies the data processing inequality, and the
method of proof relies upon the operator Jensen inequality, similar to Petz's
original approach. Interestingly, the sandwiched R\'enyi relative entropies
are particular examples of the optimized $f$-divergence. Thus, one benefit of
this paper is that there is now a single, unified approach for establishing
the data processing inequality for both the Petz--R\'enyi and sandwiched
R\'enyi relative entropies, for the full range of parameters for which it is
known to hold. This paper discusses other aspects of the optimized
$f$-divergence, such as the classical case, the classical--quantum case, and
how to construct optimized $f$-information measures.

\end{abstract}

\section{Introduction}

The quantum relative entropy \cite{U62}\ is a foundational distinguishability
measure in quantum information theory. It is a function of two quantum states
and measures how well one can tell the two states apart by a
quantum-mechanical experiment. It is well known by now to be a parent quantity
for many other information measures, such as entropy, mutual information,
conditional entropy, and entanglement measures (see, e.g., \cite{D11,W17}).
One important reason for why it has found such widespread application is that
it satisfies a data-processing inequality \cite{Lindblad1975,U77}: it does not
increase under the action of a quantum channel on the two states. This can be
interpreted as saying that two quantum states do not become more
distinguishable if the same quantum channel is applied to them, and a precise
interpretation of this statement in terms of quantum hypothesis testing is
available in \cite{HP91,ON00,BS12}. Naturally, the notion of quantum relative
entropy generalizes its classical counterpart \cite{kullback1951}, which
enjoyed a rich and illustrious history prior to the development of quantum
relative entropy.

The wide interest in relative entropy sparked various researchers to
generalize and study it further, in an attempt to elucidate the fundamental
properties that govern its behavior. One notable generalization is R\'enyi's
relative entropy \cite{Renyi61}, but this was subsequently generalized even
further in the form of the $f$-divergence \cite{C67,Ali1966,M63}. For
probability distributions $\{p(x)\}_{x}$ and $\{q(x)\}_{x}$ and a convex
function $f$, the $f$-divergence is defined as%
\begin{equation}
\sum_{x}q(x)f(p(x)/q(x)),
\end{equation}
in the case that $p(x)=0$ for all $x$ such that $q(x)=0$. The resulting
quantity is then non-increasing under the action of a classical channel
$r(y|x)$ (a conditional probability distribution), that produces the output
distributions $\sum_{x}r(y|x)p(x)$ and $\sum_{x}r(y|x)q(x)$. Some years after
these developments, a quantum generalization of $f$-divergence appeared in
\cite{P85,P86}, going under the name of \textquotedblleft
quasi-entropy\textquotedblright\ as used in \cite{Wehrl1979}. In
\cite{P85,P86} and a later development \cite{TCR09}, the quantum
data-processing inequality was proved in full generality for arbitrary quantum
channels, whenever the underlying function $f$ is \textit{operator} convex. A
relatively large literature on the topic of quantum $f$-divergence has now
developed, so much that there are now many reviews and extensions of the
original idea
\cite{OP93,PR98,NP04,PS09,TCR09,S10,P10,P10a,HMPB11,HP12,M13,HM17}.

Interestingly, when generalizing a notion from classical to quantum
information theory, there is often more than one way to do so, and sometimes
there could even be an infinite number of ways to do so. This has to do with
the non-commutativity of quantum states, and for states of many-particle
quantum systems, entanglement is involved as well. For example, there are
several different ways that one could generalize the relative entropy to the
quantum case, and two prominent formulas were put forward in \cite{U62}\ and
\cite{Belavkin1982}. This added complexity for the quantum case could
potentially be problematic, but the typical way of determining on which
generalizations we should focus is to show that a given formula is the answer
to a meaningful operational task. The papers \cite{HP91,ON00} accomplished
this for the quantum relative entropy of \cite{U62}, and since then,
researchers have realized more and more just how foundational the formula of
\cite{U62} is. As a consequence, the formula of \cite{U62} is now known as
quantum relative entropy.

The situation becomes more intricate when it comes to quantum generalizations
of R\'enyi relative entropy. For many years, the Petz--R\'enyi relative
entropy of \cite{P85,P86} has been widely studied and given an operational
interpretation \cite{N06,Hay07}, again in the context of quantum hypothesis
testing (specifically, the error exponent problem). However, in recent years,
the sandwiched R\'enyi relative entropy of \cite{MDSFT13,WWY13}\ has gained
prominence, due to its role in establishing strong converses for communication
tasks \cite{WWY13,GW15,TWW14,CMW14,DingW15,WTB16}. The result of
\cite{MO13}\ solidified its fundamental meaning in quantum information
theory:\ these authors proved that it has an operational interpretation in the
strong converse exponent of quantum hypothesis testing. As such, the situation
we are faced with is that there are two generalizations of R\'enyi relative
entropy that should be considered in quantum information theory, due to their
operational role mentioned above. There are further generalizations of the
aforementioned quantum R\'enyi relative entropies \cite{AD15}, but their
operational meaning (and thus their role in quantum information theory) is unclear.

The same work that introduced the Petz--R\'enyi relative entropy also
introduced a quantum generalization of the notion of $f$-divergence
\cite{P85,P86} (see also \cite{HMPB11}), with the Petz--R\'enyi relative
entropy being a particular example. Since then, other quantum $f$-divergences
have appeared \cite{PR98,HM17}, now known as minimal and maximal
$f$-divergences \cite{M13,HM17}. However, hitherto it has not been known how
the sandwiched R\'enyi relative entropy fits into the paradigm of quantum
$f$-divergences. In fact, the authors of \cite{HMPB11} declared in their
Example~2.11\ that a particular instance of the sandwiched R\'enyi relative
entropy is not a quantum $f$-divergence, suggesting that it would not be
possible to express it as such.

In this paper, I modify Petz's definition of quantum $f$-divergence
\cite{P85,P86,HMPB11}, by allowing for a particular optimization (see
Definition~\ref{def:opt-f-div}\ for details of the modification). As such, I
call the resulting quantity the \textit{optimized} quantum $f$-divergence. I
prove that it obeys a quantum data processing inequality, and as such, my
perspective is that it deserves to be considered as another variant of the
quantum $f$-divergence, in addition to the original, the minimal, and the
maximal. Interestingly, the sandwiched R\'enyi relative entropy is directly
related to the optimized quantum $f$-divergence, thus bringing the sandwiched
quantity into the $f$-divergence formalism.

One benefit of the results of this paper is that there is now a single,
unified approach for establishing the data-processing inequality for both the
Petz--R\'enyi relative entropy and the sandwiched R\'enyi relative entropy,
for the full R\'enyi parameter ranges for which it is known to hold. This
unified approach is based on Petz's original approach that employed the
operator Jensen inequality \cite{HP03}, which is the statement that
\begin{equation}
f(V^{\dag}X V) \leq V^{\dag}f(X) V,
\end{equation}
where $f$ is an operator convex function defined on an interval $I$, $X$ is a
Hermitian operator with spectrum in $I$, and $V$ is an isometry. This unified
approach is useful for presenting a succint proof of the data processing
inequality for both quantum R\'enyi relative entropy families.

In the rest of the paper, I begin by defining the optimized quantum
$f$-divergence and then discuss various alternative ways of writing it,
including its representation in terms of the relative modular operator
formalism. In Section~\ref{sec:DP}, I prove that the optimized $f$-divergence
satisfies the quantum data processing inequality whenever the underlying
function $f$ is operator anti-monotone with domain $(0,\infty)$ and range
$\mathbb{R}$. The proof of quantum data processing has two steps: I first
prove that the optimized quantum $f$-divergence is invariant under isometric
embeddings and then show that it is monotone non-increasing under the action
of a partial trace. By the Stinespring dilation theorem \cite{S55}, these two
steps establish data processing under general quantum channels. The core tool
underlying both steps is the operator Jensen inequality \cite{HP03}. The proof
of monotonicity under partial trace features some novel aspects for handling
non-invertible operators. In Section~\ref{sec:examples-opt-f-div}, I show how
the quantum relative entropy and the sandwiched R\'enyi relative entropies are
directly related to the optimized quantum $f$-divergence.
Section~\ref{sec:petz-f}\ then discusses the relation between Petz's
$f$-divergence and the optimized one. Section~\ref{sec:c-cq}\ shows how the
optimized $f$-divergence simplifies when the operators involved have a
classical or classical--quantum form. In Section~\ref{sec:f-info},\ I discuss
how to construct several information measures from the optimized
$f$-divergence, which could potentially find application in quantum
information theory or resource theories. I finally conclude in
Section~\ref{sec:conclusion}\ with a summary and some open directions.

\section{Optimized quantum $f$-divergence}

Let us begin by formally defining the optimized quantum $f$-divergence:

\begin{definition}
[Optimized quantum $f$-divergence]\label{def:opt-f-div}Let $f$ be a function
with domain $(0,\infty)$ and range $\mathbb{R}$. For positive semi-definite
operators $X$ and $Y$ acting on a Hilbert space $\mathcal{H}_{S}$, we define
the \textit{optimized} quantum $f$-divergence as%
\begin{equation}
\widetilde{Q}_{f}(X\Vert Y)\equiv\sup_{\tau>0,\ \operatorname{Tr}\{\tau
\}\leq1,\ \varepsilon>0}\widetilde{Q}_{f}(X\Vert Y+\varepsilon\Pi_{Y}^{\perp
};\tau), \label{eq:q-f-Y-not-invertible}%
\end{equation}
where $\widetilde{Q}_{f}(X\Vert Z;\tau)$ is defined for positive definite $Z$
and $\tau$ acting on $\mathcal{H}_{S}$ as%
\begin{align}
\widetilde{Q}_{f}(X\Vert Z;\tau)  &  \equiv\langle\varphi^{X}|_{S\hat{S}%
}f(\tau_{S}^{-1}\otimes Z_{\hat{S}}^{T})|\varphi^{X}\rangle_{S\hat{S}%
},\label{eq:q_f_tau}\\
|\varphi^{X}\rangle_{S\hat{S}}  &  \equiv(X_{S}^{1/2}\otimes I_{\hat{S}%
})|\Gamma\rangle_{S\hat{S}}.
\end{align}
In the above, $\Pi_{Y}^{\perp}$ denotes the projection onto the kernel of $Y$,
$\mathcal{H}_{\hat{S}}$ is an auxiliary Hilbert space isomorphic to
$\mathcal{H}_{S}$,%
\begin{equation}
\left\vert \Gamma\right\rangle _{S\hat{S}}\equiv\sum_{i=1}^{\left\vert
S\right\vert }\left\vert i\right\rangle _{S}\left\vert i\right\rangle
_{\hat{S}},
\end{equation}
for orthonormal bases $\{\left\vert i\right\rangle _{S}\}_{i=1}^{\left\vert
S\right\vert }$ and $\{\left\vert i\right\rangle _{\hat{S}}\}_{i=1}^{|\hat
{S}|}$, and the $T$ superscript indicates transpose with respect to the basis
$\{\left\vert i\right\rangle _{\hat{S}}\}_{i}$.
\end{definition}

\begin{remark}
Note that the expression in \eqref{eq:q-f-Y-not-invertible} simplifies
considerably in the case that $Y$ is positive definite. That is, it reduces to
the following simpler expression in the case that $Y>0$:
\begin{equation}
\widetilde{Q}_{f}(X\Vert Y) = \sup_{\tau>0,\ \operatorname{Tr}\{\tau\}\leq1}
\langle\varphi^{X}|_{S\hat{S}}f(\tau_{S}^{-1}\otimes Y_{\hat{S}}^{T}%
)|\varphi^{X}\rangle_{S\hat{S}}.
\end{equation}
As such, the optimized $f$-divergence in \eqref{eq:q-f-Y-not-invertible}
represents a modification of Petz's quantum $f$-divergence, a topic that I
discuss in more detail in Section~\ref{sec:petz-f}. The intention of the more
general definition in \eqref{eq:q-f-Y-not-invertible} is to provide a
consistent way of defining the optimized $f$-divergence in the case that $Y$
is not positive semi-definite.
\end{remark}

The case of greatest interest for us here is when the underlying function $f$
is operator anti-monotone; i.e., for Hermitian operators $A$ and $B$, the
function $f$ is such that $A\leq B\Rightarrow f(B)\leq f(A)$ (see, e.g.,
\cite{B97}). This property is rather strong, but there are several functions
of interest in quantum-physical applications that obey it (see
Section~\ref{sec:examples-opt-f-div}). One critical property of an operator
anti-monotone function with domain $(0,\infty)$ and range $\mathbb{R}$ is that
it is also operator convex and continuous (see, e.g., \cite{H13}). In this
case, we have the following proposition:

\begin{proposition}
Let $f$ be an operator anti-monotone function with domain $(0,\infty)$ and
range $\mathbb{R}$. For positive semi-definite operators $X$ and $Y$ acting on
a Hilbert space $\mathcal{H}_{S}$, the following equality holds%
\begin{equation}
\widetilde{Q}_{f}(X\Vert Y)=\sup_{\tau>0,\,\operatorname{Tr}\{\tau\}=1}%
\lim_{\varepsilon\searrow0}\widetilde{Q}_{f}(X\Vert Y+\varepsilon\Pi
_{Y}^{\perp};\tau),
\end{equation}
and furthermore, the function $\widetilde{Q}_{f}(X\Vert Y+\varepsilon\Pi
_{Y}^{\perp};\tau)$ is concave in $\tau$. Finally, for positive semi-definite
$Y_{1}$ and $Y_{2}$ such that $Y_{1}\leq Y_{2}$, we have that%
\begin{equation}
\widetilde{Q}_{f}(X\Vert Y_{1})\geq\widetilde{Q}_{f}(X\Vert Y_{2}).
\label{eq:dominating-ineq}%
\end{equation}

\end{proposition}

\begin{proof}
To see that we can restrict the optimization over $\tau$ to $\tau$ satisfying
$\operatorname{Tr}\{\tau\}=1$, let $\tau$ be such that $\tau>0$ and
$\operatorname{Tr}\{\tau\}<1$. Then%
\begin{equation}
\tau_{S}^{-1}\otimes Y_{\hat{S}}^{T}=\frac{1}{\operatorname{Tr}\{\tau_{S}%
\}}\left[  \frac{\tau_{S}}{\operatorname{Tr}\{\tau_{S}\}}\right]  ^{-1}\otimes
Y_{\hat{S}}^{T}\geq\left[  \frac{\tau_{S}}{\operatorname{Tr}\{\tau_{S}%
\}}\right]  ^{-1}\otimes Y_{\hat{S}}^{T},
\end{equation}
and so%
\begin{equation}
\widetilde{Q}_{f}(X\Vert Y+\varepsilon\Pi_{Y}^{\perp};\tau)\leq\widetilde
{Q}_{f}(X\Vert Y+\varepsilon\Pi_{Y}^{\perp};\tau/\operatorname{Tr}\{\tau\})
\end{equation}
from the operator anti-monotonicity of $f$. Changing $\sup_{\varepsilon>0}$ to
$\lim_{\varepsilon\searrow0}$ follows as well from operator anti-monotonicity
of $f$. Let $\varepsilon_{2}\geq\varepsilon_{1}>0$. Then $Y+\varepsilon_{1}%
\Pi_{Y}^{\perp}\leq Y+\varepsilon_{2}\Pi_{Y}^{\perp}$ and so $\widetilde
{Q}_{f}(X\Vert Y+\varepsilon_{2}\Pi_{Y}^{\perp};\tau)\leq\widetilde{Q}%
_{f}(X\Vert Y+\varepsilon_{1}\Pi_{Y}^{\perp};\tau)$. So then the highest value
of $\widetilde{Q}_{f}(X\Vert Y+\varepsilon\Pi_{Y}^{\perp};\tau)$ is achieved
in the limit as $\varepsilon\searrow0$, where we have also invoked the
continuity of $f$. We note that this limit could evaluate to infinity.

Concavity in $\tau$ follows because%
\begin{equation}
f(\tau_{S}^{-1}\otimes Y_{\hat{S}}^{T})=f\!\left(  \left[  \tau_{S}%
\otimes\left(  Y_{\hat{S}}^{T}\right)  ^{-1}\right]  ^{-1}\right)  ,
\end{equation}
and the function $f(x^{-1})$ is operator monotone on $(0,\infty)$, given that
it is the composition of the operator anti-monotone function $x^{-1}$ with
domain $(0,\infty)$ and range $(0,\infty)$ and the function $f$, taken to be
operator anti-monotone on $(0,\infty)$ by hypothesis. Since $f(x^{-1})$ is
operator monotone on $(0,\infty)$, it is operator concave (see, e.g.,
\cite{H13}).

The dominating property in \eqref{eq:dominating-ineq} follows from the fact
that $f$ is operator anti-monotone on $(0,\infty)$, which implies the
following for a fixed $\varepsilon>0$ and $\tau_{S}$ such that $\tau_{S}>0$
and $\operatorname{Tr}\{\tau_{S}\}\leq1$:%
\begin{equation}
f(\tau_{S}^{-1}\otimes(Y_{1}+\varepsilon\Pi_{Y_{1}}^{\perp})_{\hat{S}}%
^{T})\geq f(\tau_{S}^{-1}\otimes(Y_{2}+\varepsilon\Pi_{Y_{1}}^{\perp}%
)_{\hat{S}}^{T}).
\end{equation}
We arrive at the inequality in \eqref{eq:dominating-ineq} after sandwiching by
$|\varphi^{X}\rangle_{S\hat{S}}$, taking the limit as $\varepsilon\searrow0$,
and taking a supremum over $\tau$.
\end{proof}

\bigskip

For $X$ positive semi-definite and $Y$ and $\tau$ positive definite, with
spectral decompositions of $Y$ and $\tau$ given as%
\begin{equation}
Y=\sum_{y}\mu_{y}|\phi^{y}\rangle\langle\phi^{y}|,\qquad\tau=\sum_{t}\nu
_{t}|\psi^{t}\rangle\langle\psi^{t}|,
\end{equation}
we can write%
\begin{align}
\widetilde{Q}_{f}(X\Vert Y;\tau)  &  =\sum_{y,t}f(\mu_{y}\nu_{t}%
^{-1})\operatorname{Tr}\{X^{1/2}|\phi^{y}\rangle\langle\phi^{y}|X^{1/2}%
|\psi^{t}\rangle\langle\psi^{t}|\}\label{eq:f-divergence-spec-form}\\
&  =\sum_{y,t}f(\mu_{y}\nu_{t}^{-1})|\langle\phi^{y}|X^{1/2}|\psi^{t}%
\rangle|^{2},
\end{align}
by using the facts that%
\begin{align}
f(\tau_{S}^{-1}\otimes Y_{\hat{S}}^{T})  &  =\sum_{y,t}f(\mu_{y}\nu_{t}%
^{-1})|\psi^{t}\rangle\langle\psi^{t}|_{S}\otimes|\phi^{y}\rangle\langle
\phi^{y}|_{\hat{S}}^{T},\\
(I_{S}\otimes Z_{\hat{S}}^{T})|\Gamma\rangle_{S\hat{S}}  &  =(Z_{S}\otimes
I_{\hat{S}})|\Gamma\rangle_{S\hat{S}},\label{eq:transpose-trick}\\
\langle\Gamma|_{S\hat{S}}(Z_{S}\otimes I_{\hat{S}})|\Gamma\rangle_{S\hat{S}}
&  =\operatorname{Tr}\{Z_{S}\}, \label{eq:gamma-to-trace}%
\end{align}
for any square operator $Z$ acting on $\mathcal{H}_{S}$. The formula in
\eqref{eq:f-divergence-spec-form}\ is helpful in some parts of our analysis below.

We can also phrase Definition~\ref{def:opt-f-div} in terms of the relative
modular operator formalism, which is employed in many of the works on
quasi-entropy (many details of this formalism in the context of
quasi-entropies are available in \cite{HMPB11}). Let $P$ be a positive
semi-definite operator, and let $R$ be a positive definite operator. Defining
the action of the relative modular operator $\Delta(P/R)$\ on an operator $X$
as%
\begin{equation}
\Delta(P/R)(X)=PXR^{-1},
\end{equation}
and the Hilbert--Schmidt inner product $\left\langle W,Z\right\rangle
=\operatorname{Tr}\{W^{\dag}Z\}$, we can write the quantity $\widetilde{Q}%
_{f}(X\Vert Y;\tau)$ underlying $\widetilde{Q}_{f}(X\Vert Y)$ in terms of the
relative modular operator as%
\begin{equation}
\widetilde{Q}_{f}(X\Vert Y;\tau)=\langle X^{1/2},f(\Delta(Y/\tau
))(X^{1/2})\rangle. \label{eq:rel-mod-op}%
\end{equation}
The definition of optimized quantum $f$-divergence following from plugging
\eqref{eq:rel-mod-op} into \eqref{eq:q-f-Y-not-invertible} can be used in more
general contexts than those considered in the present paper (for example, in
the context of von Neumann algebras). However, in this work, we find it more
convenient to work with the expression in \eqref{eq:q_f_tau} (see
\cite{TCR09,S10} for a similar approach), and throughout this paper, we work
in the setting of finite-dimensional quantum systems.

\section{Quantum data processing}

\label{sec:DP}Our first main objective is to prove that $\widetilde{Q}%
_{f}(X\Vert Y)$ deserves the name \textquotedblleft$f$%
-divergence\textquotedblright\ or \textquotedblleft$f$-relative
entropy,\textquotedblright\ i.e., that it is monotone non-increasing under the
action of a completely positive, trace-preserving map $\mathcal{N}$:%
\begin{equation}
\widetilde{Q}_{f}(X\Vert Y)\geq\widetilde{Q}_{f}(\mathcal{N}(X)\Vert
\mathcal{N}(Y)). \label{eq:mono-DP}%
\end{equation}
Such a map $\mathcal{N}$ is also called a quantum channel, due to its purpose
in quantum physics as modeling the physical evolution of the state of a
quantum system. In quantum information-theoretic contexts, the inequality in
\eqref{eq:mono-DP} is known as the quantum data processing inequality.
According to the Stinespring dilation theorem \cite{S55}, every quantum
channel can be realized by an isometric embedding of its input into a tensor
product of the channel's output Hilbert space and an auxiliary Hilbert space,
followed by a partial trace over the auxiliary Hilbert space. That is, to
every quantum channel $\mathcal{N}_{S\rightarrow B}$, there exists an isometry
$U_{S\rightarrow BE}^{\mathcal{N}}$ such that%
\begin{equation}
\mathcal{N}_{S\rightarrow B}(X_{S})=\operatorname{Tr}_{E}\{U_{S\rightarrow
BE}^{\mathcal{N}}X_{S}\left(  U_{S\rightarrow BE}^{\mathcal{N}}\right)
^{\dag}\}.
\end{equation}
As such, we can prove the inequality in \eqref{eq:mono-DP} in two steps:

\begin{enumerate}
\item \textit{Isometric invariance}:\ First show that%
\begin{equation}
\widetilde{Q}_{f}(X\Vert Y)=\widetilde{Q}_{f}(UXU^{\dag}\Vert UYU^{\dag})
\end{equation}
for any isometry $U$ and any positive semi-definite $X$ and $Y$.\footnote{The importance of establishing isometric invariance of
quasi-entropies has been stressed in \cite{TCR09} and \cite[Appendix~B]{T12}.}

\item \textit{Monotonicity under partial trace}:\ Then show that%
\begin{equation}
\widetilde{Q}_{f}(X_{AB}\Vert Y_{AB})\geq\widetilde{Q}_{f}(X_{A}\Vert Y_{A})
\end{equation}
for positive semi-definite operators $X_{AB}$ and $Y_{AB}$ acting on the
tensor-product Hilbert space $\mathcal{H}_{A}\otimes\mathcal{H}_{B}$, with
$X_{A}=\operatorname{Tr}_{B}\{X_{AB}\}$ and $Y_{A}=\operatorname{Tr}%
_{B}\{Y_{AB}\}$.
\end{enumerate}

\noindent So we proceed and first prove isometric invariance:

\begin{proposition}
[Isometric invariance]\label{prop:iso-inv}Let $U:\mathcal{H}_{S}%
\rightarrow\mathcal{H}_{R}$ be an isometry, let $X$ and $Y$ be positive
semi-definite operators, and let $f$ be an operator anti-monotone function
with domain $(0,\infty)$ and range $\mathbb{R}$. Then the following equality
holds%
\begin{equation}
\widetilde{Q}_{f}(X\Vert Y)=\widetilde{Q}_{f}(UXU^{\dag}\Vert UYU^{\dag}).
\end{equation}

\end{proposition}

\begin{proof}
In the case that $\dim(\mathcal{H}_{S})=\dim(\mathcal{H}_{R})$, the statement
holds trivially because $U$ is a unitary and then $\mathcal{H}_{S}$ and
$\mathcal{H}_{R}$ are isomorphic. So we focus on the case in which
$\dim(\mathcal{H}_{S})<\dim(\mathcal{H}_{R})$. First suppose that $Y$ is
invertible when acting on $\mathcal{H}_{S}$. The operator $X$ is generally not
invertible, and with respect to the decomposition of $\mathcal{H}_{S}$ as
$\operatorname{supp}(X)\oplus\ker(X)$, we can write $X$ and each
eigenprojection $|\phi^{y}\rangle\langle\phi^{y}|$\ of $Y$ respectively as%
\begin{equation}%
\begin{bmatrix}
X & 0\\
0 & 0
\end{bmatrix}
,\qquad%
\begin{bmatrix}
\phi_{11}^{y} & \phi_{12}^{y}\\
\phi_{21}^{y} & \phi_{22}^{y}%
\end{bmatrix}
.
\end{equation}
Let $\tau$ acting on $\mathcal{H}_{S}$ be such that $\tau>0$ and
$\operatorname{Tr}\{\tau\}=1$. Suppose that its spectral decomposition is
given by $\sum_{t=1}^{\left\vert S\right\vert }\nu_{t}|\psi^{t}\rangle
\langle\psi^{t}|$, with each $\nu_{t}\in(0,1)$ and $|\psi^{t}\rangle$ a unit
vector such that $\sum_{t}\nu_{t}=1$. We can then write each eigenprojection
$|\psi^{t}\rangle\langle\psi^{t}|$ with respect to the decomposition of
$\mathcal{H}_{S}$ as $\operatorname{supp}(X)\oplus\ker(X)$ as%
\begin{equation}%
\begin{bmatrix}
\psi_{11}^{t} & \psi_{12}^{t}\\
\psi_{21}^{t} & \psi_{22}^{t}%
\end{bmatrix}
.
\end{equation}
Applying definitions and \eqref{eq:f-divergence-spec-form}, we then find that%
\begin{align}
\widetilde{Q}_{f}(X\Vert Y;\tau)  &  =\sum_{y,t}f(\mu_{y}\nu_{t}%
^{-1})\operatorname{Tr}\left\{
\begin{bmatrix}
\sqrt{X} & 0\\
0 & 0
\end{bmatrix}%
\begin{bmatrix}
\phi_{11}^{y} & \phi_{12}^{y}\\
\phi_{21}^{y} & \phi_{22}^{y}%
\end{bmatrix}%
\begin{bmatrix}
\sqrt{X} & 0\\
0 & 0
\end{bmatrix}%
\begin{bmatrix}
\psi_{11}^{t} & \psi_{12}^{t}\\
\psi_{21}^{t} & \psi_{22}^{t}%
\end{bmatrix}
\right\} \\
&  =\sum_{y,t}f(\mu_{y}\nu_{t}^{-1})\operatorname{Tr}\{\sqrt{X}\phi_{11}%
^{y}\sqrt{X}\psi_{11}^{t}\}. \label{eq:smaller-Hilbert-f-div}%
\end{align}

Now consider $\widetilde{Q}_{f}$\ for $UXU^{\dag}$ and $UYU^{\dag}$. Without
loss of generality, we can consider the isometry $U$ to be the trivial
embedding of $\mathcal{H}_{S}$ into the larger Hilbert space $\mathcal{H}_{R}%
$, with it decomposed as $\mathcal{H}_{R}=\mathcal{H}_{S}\oplus\mathcal{H}%
_{S}^{\perp}$, so that with respect to the decomposition of $\mathcal{H}_{R}$
as $\mathcal{H}_{R}=\operatorname{supp}(X)\oplus\ker(X)\oplus\mathcal{H}%
_{S}^{\perp}$, we can write $X$ and each eigenprojection $|\phi^{y}%
\rangle\langle\phi^{y}|$\ of $Y$ in the larger Hilbert space $\mathcal{H}_{R}%
$\ as%
\begin{equation}%
\begin{bmatrix}
X & 0 & 0\\
0 & 0 & 0\\
0 & 0 & 0
\end{bmatrix}
,\qquad%
\begin{bmatrix}
\phi_{11}^{y} & \phi_{12}^{y} & 0\\
\phi_{21}^{y} & \phi_{22}^{y} & 0\\
0 & 0 & 0
\end{bmatrix}
.
\end{equation}
We use the notation $X_R$ to denote the operator $X$ embedded into $\mathcal{H}_R$.
Let $\omega$ acting on $\mathcal{H}_{R}$ be such that $\omega>0$ and
$\operatorname{Tr}\{\omega\}=1$. Suppose that its spectral decomposition is
given by $\sum_{s=1}^{\left\vert R\right\vert }\lambda_{s}|\varphi^{s}%
\rangle\langle\varphi^{s}|$, with each $\lambda_{s}\in(0,1)$ and $|\varphi
^{s}\rangle$ a unit vector such that $\sum_{s}\lambda_{s}=1$. We can then
write each eigenprojection $|\varphi^{s}\rangle\langle\varphi^{s}|$ as%
\begin{equation}%
\begin{bmatrix}
\varphi_{11}^{s} & \varphi_{12}^{s} & \varphi_{13}^{s}\\
\varphi_{21}^{s} & \varphi_{22}^{s} & \varphi_{23}^{s}\\
\varphi_{31}^{s} & \varphi_{32}^{s} & \varphi_{33}^{s}%
\end{bmatrix}
.
\end{equation}
Since $Y$ is no longer invertible after the embedding, we need to instead
consider the operator $Y_R+\varepsilon\Pi_{Y}^{\perp}$ for some $\varepsilon
\in(0,1)$, where we use the notation $Y_R$ to denote the operator $Y$ embedded into $\mathcal{H}_R$. Then the eigenprojections of $Y_R+\varepsilon\Pi_{Y}^{\perp}$ are now
represented in this larger space as%
\begin{equation}%
\begin{bmatrix}
\phi_{11}^{y} & \phi_{12}^{y} & 0\\
\phi_{21}^{y} & \phi_{22}^{y} & 0\\
0 & 0 & 0
\end{bmatrix}
,\qquad%
\begin{bmatrix}
0 & 0 & 0\\
0 & 0 & 0\\
0 & 0 & I
\end{bmatrix}
.
\end{equation}
Applying definitions, we then find that, in the larger Hilbert space
$\mathcal{H}_{R}$,%
\begin{align}
&  \widetilde{Q}_{f}(X_R\Vert Y_R+\varepsilon\Pi_{Y}^{\perp};\omega)\nonumber\\
&  =\sum_{y,s}f(\mu_{y}\lambda_{s}^{-1})\operatorname{Tr}\left\{
\begin{bmatrix}
\sqrt{X} & 0 & 0\\
0 & 0 & 0\\
0 & 0 & 0
\end{bmatrix}%
\begin{bmatrix}
\phi_{11}^{y} & \phi_{12}^{y} & 0\\
\phi_{21}^{y} & \phi_{22}^{y} & 0\\
0 & 0 & 0
\end{bmatrix}%
\begin{bmatrix}
\sqrt{X} & 0 & 0\\
0 & 0 & 0\\
0 & 0 & 0
\end{bmatrix}%
\begin{bmatrix}
\varphi_{11}^{s} & \varphi_{12}^{s} & \varphi_{13}^{s}\\
\varphi_{21}^{s} & \varphi_{22}^{s} & \varphi_{23}^{s}\\
\varphi_{31}^{s} & \varphi_{32}^{s} & \varphi_{33}^{s}%
\end{bmatrix}
\right\} \nonumber\\
&  \qquad+\sum_{t^{\prime}}f(\varepsilon\lambda_{s}^{-1})\operatorname{Tr}%
\left\{
\begin{bmatrix}
\sqrt{X} & 0 & 0\\
0 & 0 & 0\\
0 & 0 & 0
\end{bmatrix}%
\begin{bmatrix}
0 & 0 & 0\\
0 & 0 & 0\\
0 & 0 & I
\end{bmatrix}%
\begin{bmatrix}
\sqrt{X} & 0 & 0\\
0 & 0 & 0\\
0 & 0 & 0
\end{bmatrix}%
\begin{bmatrix}
\varphi_{11}^{s} & \varphi_{12}^{s} & \varphi_{13}^{s}\\
\varphi_{21}^{s} & \varphi_{22}^{s} & \varphi_{23}^{s}\\
\varphi_{31}^{s} & \varphi_{32}^{s} & \varphi_{33}^{s}%
\end{bmatrix}
\right\} \\
&  =\sum_{y,s}f(\mu_{y}\lambda_{s}^{-1})\operatorname{Tr}\{\sqrt{X}\phi
_{11}^{y}\sqrt{X}\varphi_{11}^{s}\}. \label{eq:larger-Hilbert-f-div}%
\end{align}
We now compare the expressions in \eqref{eq:smaller-Hilbert-f-div} and
\eqref{eq:larger-Hilbert-f-div}. For a given $\tau>0$ with spectral
decomposition $\sum_{t=1}^{\left\vert S\right\vert }\nu_{t}|\psi^{t}%
\rangle\langle\psi^{t}|$ and $\delta\in(0,1)$, we can choose $\omega
(\delta)>0$ as%
\begin{equation}
\omega(\delta)=\left(  1-\delta\right)  \sum_{t}\nu_{t}%
\begin{bmatrix}
\psi_{11}^{t} & \psi_{12}^{t} & 0\\
\psi_{21}^{t} & \psi_{22}^{t} & 0\\
0 & 0 & 0
\end{bmatrix}
+\delta%
\begin{bmatrix}
0 & 0 & 0\\
0 & 0 & 0\\
0 & 0 & I/\dim(\mathcal{H}_{S}^{\perp})
\end{bmatrix}
,
\end{equation}
and then, by the above reasoning, we have that%
\begin{equation}
\widetilde{Q}_{f}(X_R\Vert Y_R+\varepsilon\Pi_{Y}^{\perp};\omega(\delta
))=\sum_{y,t}f(\mu_{y}\left[  (1-\delta)\nu_{t}\right]  ^{-1}%
)\operatorname{Tr}\{\sqrt{X}\phi_{11}^{y}\sqrt{X}\psi_{11}^{t}\}.
\end{equation}
Taking the limit $\delta\searrow0$ and applying the continuity of $f$ then
gives%
\begin{equation}
\lim_{\delta\searrow0}\widetilde{Q}_{f}(X_R\Vert Y_R+\varepsilon\Pi_{Y}^{\perp
};\omega(\delta))=\widetilde{Q}_{f}(X\Vert Y;\tau).
\end{equation}
So it is clear that the following inequality holds for all $\tau$:%
\begin{equation}
\widetilde{Q}_{f}(X\Vert Y;\tau)\leq\sup_{\omega>0,\ \operatorname{Tr}%
\{\omega\}=1}\widetilde{Q}_{f}(X_R\Vert Y_R+\varepsilon\Pi_{Y}^{\perp};\omega).
\end{equation}
We can thus conclude that%
\begin{equation}
\sup_{\tau>0,\ \operatorname{Tr}\{\tau\}=1}\widetilde{Q}_{f}(X\Vert
Y;\tau)\leq\sup_{\omega>0,\ \operatorname{Tr}\{\omega\}=1}\widetilde{Q}%
_{f}(X_R\Vert Y_R+\varepsilon\Pi_{Y}^{\perp};\omega),
\end{equation}
which is the same as the inequality%
\begin{equation}
\widetilde{Q}_{f}(X\Vert Y)\leq\widetilde{Q}_{f}(UXU^{\dag}\Vert UYU^{\dag}).
\end{equation}
This establishes the inequality $\widetilde{Q}_{f}(X\Vert Y)\leq\widetilde
{Q}_{f}(UXU^{\dag}\Vert UYU^{\dag})$ in the case in which $Y$ is invertible
when acting on $\mathcal{H}_{S}$.

Given that the function $x^{-1}$ is operator anti-monotone on $(0,\infty)$ and
has range $(0,\infty)$, it follows that $f(x^{-1})=g(x)$ is operator monotone
on $(0,\infty)$ and thus operator concave \cite{H13}. Defining the embedding
isometry $V\equiv V_{S\rightarrow R}\otimes V_{\hat{S}\rightarrow\hat{R}%
}\equiv\sum_{i}|i\rangle_{R}\langle i|_{S}\otimes\sum_{j}|j\rangle_{\hat{R}%
}\langle j|_{\hat{S}}$, we then have by a direct application of the operator
Jensen inequality \cite{HP03} and the fact that $g$ is operator concave\ that%
\begin{align}
\widetilde{Q}_{f}(X\Vert Y+\varepsilon\Pi_{Y}^{\perp};\omega)  &
=\langle\varphi^{X}|_{R\hat{R}}f(\omega_{R}^{-1}\otimes\left[  Y_{\hat{R}%
}+\varepsilon\Pi_{Y}^{\perp}\right]  ^{T})|\varphi^{X}\rangle_{R\hat{R}}\\
&  =\langle\varphi^{X}|_{S\hat{S}}V^{\dag}f\!\left(  \left[  \omega_{R}%
\otimes\left[  Y_{\hat{R}}^{-1}+\varepsilon^{-1}\Pi_{Y}^{\perp}\right]
^{T}\right]  ^{-1}\right)  V|\varphi^{X}\rangle_{S\hat{S}}\\
&  =\langle\varphi^{X}|_{S\hat{S}}V^{\dag}g\!\left(  \omega_{R}\otimes\left[
Y_{\hat{R}}^{-1}+\varepsilon^{-1}\Pi_{Y}^{\perp}\right]  ^{T}\right)
V|\varphi^{X}\rangle_{S\hat{S}}\\
&  \leq\langle\varphi^{X}|_{S\hat{S}} g\!\left(  V^{\dag}\left[  \omega
_{R}\otimes\left[  Y_{\hat{R}}^{-1}+\varepsilon^{-1}\Pi_{Y}^{\perp}\right]
^{T}\right]  V\right)  |\varphi^{X}\rangle_{S\hat{S}}\\
&  =\langle\varphi^{X}|_{S\hat{S}} g\!\left(  \omega_{S}^{\prime}%
\otimes\left[  Y_{\hat{S}}^{-1}\right]  ^{T}\right)  |\varphi^{X}%
\rangle_{S\hat{S}}\\
&  =\langle\varphi^{X}|_{S\hat{S}}f(\left(  \omega_{S}^{\prime}\right)
^{-1}\otimes Y_{\hat{S}}^{T})|\varphi^{X}\rangle_{S\hat{S}}\\
&  =\widetilde{Q}_{f}(X\Vert Y;\omega_{S}^{\prime})\\
&  \leq\widetilde{Q}_{f}(X\Vert Y),
\end{align}
where $\omega_{S}^{\prime}\equiv\left(  V_{S\rightarrow R}\right)  ^{\dag
}\omega_{R}V_{S\rightarrow R}$ is an operator acting on $\mathcal{H}_{S}$ such
that $\omega_{S}^{\prime}>0$ and $\operatorname{Tr}\{\omega_{S}^{\prime}%
\}\leq1$. In the above, the notation $Y_{\hat{R}}^{-1}$ indicates the inverse
on the support of $Y_{\hat{R}}^{-1}$, and we have employed the facts that
$\left[  Y_{\hat{R}}+\varepsilon\Pi_{Y}^{\perp}\right]  ^{-1}=Y_{\hat{R}}%
^{-1}+\varepsilon^{-1}\Pi_{Y}^{\perp}$ and $\left(  V_{\hat{S}\rightarrow
\hat{R}}\right)  ^{\dag}\left(  Y_{\hat{R}}^{-1}+\varepsilon^{-1}\Pi
_{Y}^{\perp}\right)  V_{\hat{S}\rightarrow\hat{R}}=Y_{\hat{R}}^{-1}$. Since
the inequality holds for all $\omega>0$ such that $\operatorname{Tr}%
\{\omega\}=1$, we conclude that%
\begin{equation}
\widetilde{Q}_{f}(UXU^{\dag}\Vert UYU^{\dag})\leq\widetilde{Q}_{f}(X\Vert Y).
\end{equation}
We have now established the claim for invertible $Y$.

If $Y$ is not invertible when acting on $\mathcal{H}_{S}$, then the definition
in \eqref{eq:q-f-Y-not-invertible} applies, which in fact forces $Y$ to become
invertible when acting on $\mathcal{H}_{S}$. So then, in this case, we can
conclude that%
\begin{equation}
\widetilde{Q}_{f}(X\Vert Y+\varepsilon\Pi_{Y}^{\perp})=\widetilde{Q}%
_{f}(UXU^{\dag}\Vert U\left[  Y+\varepsilon\Pi_{Y}^{\perp}\right]  U^{\dag
}+\varepsilon\Pi_{\mathcal{H}_{S}^{\perp}}).
\end{equation}
So the quantities are the same for all $\varepsilon\in(0,1)$, and then the
equality follows by taking a supremum over $\varepsilon>0$.
\end{proof}

\begin{remark}
The above proof establishes the inequality $\widetilde{Q}_{f}(X\Vert
Y)\leq\widetilde{Q}_{f}(UXU^{\dag}\Vert UYU^{\dag})$ for any continuous
function $f$ with domain $(0,\infty)$ and range $\mathbb{R}$, but for the
opposite inequality $\widetilde{Q}_{f}(UXU^{\dag}\Vert UYU^{\dag}%
)\leq\widetilde{Q}_{f}(X\Vert Y)$, the proof given requires $f$ to be operator
anti-monotone with domain $(0,\infty)$ and range~$\mathbb{R}$.
\end{remark}

We now complete the second step toward quantum data processing, as mentioned above:

\begin{proposition}
[Monotonicity under partial trace]\label{prop:mono-PT}Given positive
semi-definite operators $X_{AB}$ and $Y_{AB}$ acting on the tensor-product
Hilbert space $\mathcal{H}_{A}\otimes\mathcal{H}_{B}$, the optimized quantum
$f$-divergence does not increase under the action of a partial trace, in the
sense that%
\begin{equation}
\widetilde{Q}_{f}(X_{AB}\Vert Y_{AB})\geq\widetilde{Q}_{f}(X_{A}\Vert Y_{A}),
\label{eq:DP-partial-trace}%
\end{equation}
where $X_{A}=\operatorname{Tr}_{B}\{X_{AB}\}$ and $Y_{A}=\operatorname{Tr}%
_{B}\{Y_{AB}\}$.
\end{proposition}

\begin{proof}
Throughout the proof, we take $Y_{AB}$ to be invertible on $\mathcal{H}%
_{A}\otimes\mathcal{H}_{B}$. We can do so because the supremum over
$\varepsilon>0$ can be placed on the very outside, as in
Definition~\ref{def:opt-f-div}, and then we can optimize over $\varepsilon>0$
at the very end once the monotonicity inequality has been established. There
are three cases to consider:

\begin{enumerate}
\item when $X_{AB}>0$,

\item when $X_{A}>0$, but $X_{AB}$ is not invertible, and

\item when $X_{A}$ is not invertible.
\end{enumerate}

\noindent Here I show a proof for the first two cases, and the last case is
shown in detail in the appendix for the interested reader.

\vspace{.1in}

\textbf{Case }$X_{AB}>0$\textbf{:} We establish the claim when $X_{AB}$ is
invertible, and so $X_{A}$ is as well. This is the simplest case to consider
and thus has the most transparent proof (it is fruitful to understand this
case well before considering the other cases). The
quantities of interest are as follows:%
\begin{align}
\widetilde{Q}_{f}(X_{AB}\Vert Y_{AB};\tau_{AB}) &  =\langle\varphi^{X_{AB}%
}|_{AB\hat{A}\hat{B}}f(\tau_{AB}^{-1}\otimes Y_{\hat{A}\hat{B}}^{T}%
)|\varphi^{X_{AB}}\rangle_{AB\hat{A}\hat{B}},\\
\widetilde{Q}_{f}(X_{A}\Vert Y_{A};\omega_{A}) &  =\langle\varphi^{X_{A}%
}|_{A\hat{A}}f(\omega_{A}^{-1}\otimes Y_{\hat{A}}^{T})|\varphi^{X_{A}}%
\rangle_{A\hat{A}},
\end{align}
where $\tau_{AB}$ and $\omega_{A}$ are invertible density operators and, by
definition,%
\begin{equation}
|\varphi^{X_{AB}}\rangle_{AB\hat{A}\hat{B}}=\left(  X_{AB}^{1/2}\otimes
I_{\hat{A}\hat{B}}\right)  |\Gamma\rangle_{A\hat{A}}\otimes|\Gamma
\rangle_{B\hat{B}}.
\end{equation}
The following map, acting on an operator $Z_{A}$, is a quantum channel known
as the Petz recovery channel \cite{Petz1986,Petz1988} (see also
\cite{BK02,HJPW04,LS13}):%
\begin{equation}
Z_{A}\rightarrow X_{AB}^{1/2}\left(  \left[  X_{A}^{-1/2}Z_{A}X_{A}%
^{-1/2}\right]  \otimes I_{B}\right)  X_{AB}^{1/2}.\label{eq:petz-recovery}%
\end{equation}
It is completely positive because it consists of the serial concatenation of
three completely positive maps:\ sandwiching by $X_{A}^{-1/2}$, tensoring in
the identity $I_{B}$, and sandwiching by $X_{AB}^{1/2}$. It is trace
preserving because%
\begin{align}
\operatorname{Tr}\left\{  X_{AB}^{1/2}\left(  \left[  X_{A}^{-1/2}Z_{A}%
X_{A}^{-1/2}\right]  \otimes I_{B}\right)  X_{AB}^{1/2}\right\}   &
=\operatorname{Tr}\left\{  X_{AB}\left(  \left[  X_{A}^{-1/2}Z_{A}X_{A}%
^{-1/2}\right]  \otimes I_{B}\right)  \right\}  \\
&  =\operatorname{Tr}\left\{  X_{A}\left[  X_{A}^{-1/2}Z_{A}X_{A}%
^{-1/2}\right]  \right\}  \\
&  =\operatorname{Tr}\left\{  X_{A}^{-1/2}X_{A}X_{A}^{-1/2}Z_{A}\right\}  \\
&  =\operatorname{Tr}\{Z_{A}\}.
\end{align}
The Petz recovery channel has the property that it perfectly recovers $X_{AB}$
if $X_{A}$ is input because%
\begin{equation}
X_{A}\rightarrow X_{AB}^{1/2}\left(  \left[  X_{A}^{-1/2}X_{A}X_{A}%
^{-1/2}\right]  \otimes I_{B}\right)  X_{AB}^{1/2}=X_{AB}%
.\label{eq:perfect-recovery}%
\end{equation}
Every completely positive and trace preserving map $\mathcal{N}$ has a Kraus
decomposition, which is a set $\{K_{i}\}_{i}$ of operators such that%
\begin{equation}
\mathcal{N}(\cdot)=\sum_{i}K_{i}(\cdot)K_{i}^{\dag},\qquad\sum_{i}K_{i}^{\dag
}K_{i}=I.
\end{equation}
A standard construction for an isometric extension of a channel is then to
pick an orthonormal basis $\{|i\rangle_{E}\}_{i}$ for an auxiliary Hilbert
space $\mathcal{H}_{E}$\ and define%
\begin{equation}
V=\sum_{i}K_{i}\otimes|i\rangle_{E}.\label{eq:iso-from-kraus}%
\end{equation}
One can then readily check that $\mathcal{N}(\cdot)=\operatorname{Tr}%
_{E}\{V(\cdot)V^{\dag}\}$ and $V^{\dag}V=I$. (See, e.g., \cite{W17} for a
review of these standard notions.) For the Petz recovery channel, we can
figure out a Kraus decomposition by expanding the identity operator
$I_{B}=\sum_{j=1}^{\left\vert B\right\vert }|j\rangle\langle j|_{B}$, with
respect to some orthonormal basis $\{|j\rangle_{B}\}_{j}$, so that%
\begin{align}
X_{AB}^{1/2}\left(  \left[  X_{A}^{-1/2}Z_{A}X_{A}^{-1/2}\right]  \otimes
I_{B}\right)  X_{AB}^{1/2} &  =\sum_{j=1}^{\left\vert B\right\vert }%
X_{AB}^{1/2}\left(  \left[  X_{A}^{-1/2}Z_{A}X_{A}^{-1/2}\right]
\otimes|j\rangle\langle j|_{B}\right)  X_{AB}^{1/2}\\
&  =\sum_{j=1}^{\left\vert B\right\vert }X_{AB}^{1/2}\left[  X_{A}%
^{-1/2}\otimes|j\rangle_{B}\right]  Z_{A}\left[  X_{A}^{-1/2}\otimes\langle
j|_{B}\right]  X_{AB}^{1/2}.
\end{align}
Thus, Kraus operators for the Petz recovery channel are given by%
\begin{equation}
\left\{  X_{AB}^{1/2}\left[  X_{A}^{-1/2}\otimes|j\rangle_{B}\right]
\right\}  _{j=1}^{\left\vert B\right\vert }.
\end{equation}
According to the standard recipe in \eqref{eq:iso-from-kraus}, we can
construct an isometric extension of the Petz recovery channel\ as%
\begin{align}
\sum_{j=1}^{\left\vert B\right\vert }X_{AB}^{1/2}\left[  X_{A}^{-1/2}%
\otimes|j\rangle_{B}\right]  |j\rangle_{\hat{B}} &  =X_{AB}^{1/2}X_{A}%
^{-1/2}\sum_{j=1}^{\left\vert B\right\vert }|j\rangle_{B}|j\rangle_{\hat{B}}\\
&  =X_{AB}^{1/2}X_{A}^{-1/2}|\Gamma\rangle_{B\hat{B}}.
\end{align}
We can then extend this isometry to act as an isometry on a larger space by
tensoring it with the identity operator $I_{\hat{A}}$, and so we define%
\begin{equation}
V_{A\hat{A}\rightarrow A\hat{A}B\hat{B}}\equiv X_{AB}^{1/2}\left[
X_{A}^{-1/2}\otimes I_{\hat{A}}\right]  |\Gamma\rangle_{B\hat{B}}.
\end{equation}
%The fact that $V_{A\hat{A}\rightarrow A\hat{A}B\hat{B}}$ is an isometry
%follows from the above reasoning, or equivalently because%
%\begin{align}
%V^{\dag}V &  =\left(  \langle\Gamma|_{B\hat{B}}\left[  X_{A}^{-1/2}\otimes
%I_{\hat{A}}\right]  X_{AB}^{1/2}\right)  \left(  X_{AB}^{1/2}\left[
%X_{A}^{-1/2}\otimes I_{\hat{A}}\right]  |\Gamma\rangle_{B\hat{B}}\right)  \\
%&  =\langle\Gamma|_{B\hat{B}}\left[  X_{A}^{-1/2}\otimes I_{\hat{A}}\right]
%X_{AB}\left[  X_{A}^{-1/2}\otimes I_{\hat{A}}\right]  |\Gamma\rangle_{B\hat
%{B}}\\
%&  =\left[  X_{A}^{-1/2}\otimes I_{\hat{A}}\right]  \langle\Gamma|_{B\hat{B}%
%}X_{AB}|\Gamma\rangle_{B\hat{B}}\left[  X_{A}^{-1/2}\otimes I_{\hat{A}%
%}\right]  \\
%&  =\left[  X_{A}^{-1/2}\otimes I_{\hat{A}}\right]  X_{A}\left[  X_{A}%
%^{-1/2}\otimes I_{\hat{A}}\right]  \\
%&  =X_{A}^{-1/2}X_{A}X_{A}^{-1/2}\otimes I_{\hat{A}}\\
%&  =I_{A}\otimes I_{\hat{A}}.
%\end{align}
We can also see that $V_{A\hat{A}\rightarrow A\hat{A}B\hat{B}}$ acting on
$|\varphi^{X_{A}}\rangle_{A\hat{A}}$ generates $|\varphi^{X_{AB}}%
\rangle_{AB\hat{A}\hat{B}}$:%
\begin{equation}
|\varphi^{X_{AB}}\rangle_{AB\hat{A}\hat{B}}=V_{A\hat{A}\rightarrow A\hat
{A}B\hat{B}}|\varphi^{X_{A}}\rangle_{A\hat{A}}.
\end{equation}
This can be interpreted as a generalization of \eqref{eq:perfect-recovery} in
the language of quantum information:\ an isometric extension of the Petz
recovery channel perfectly recovers a purification $|\varphi^{X_{AB}}%
\rangle_{AB\hat{A}\hat{B}}$ of $X_{AB}$ from a purification $|\varphi^{X_{A}%
}\rangle_{A\hat{A}}$ of $X_{A}$. Since the Petz recovery channel is indeed a
channel, we can pick $\tau_{AB}$ as the output state of the Petz recovery
channel acting on an invertible state $\omega_{A}$:%
\begin{equation}
\tau_{AB}=X_{AB}^{1/2}\left(  \left[  X_{A}^{-1/2}\omega_{A}X_{A}%
^{-1/2}\right]  \otimes I_{B}\right)  X_{AB}^{1/2}.\label{eq:tau-state-choice}%
\end{equation}
Observe that $\tau_{AB}$ is invertible. Then consider that%
\begin{align}
&  V^{\dag}\left(  \tau_{AB}^{-1}\otimes Y_{\hat{A}\hat{B}}^{T}\right)
V\nonumber\\
&  =\left(  \langle\Gamma|_{B\hat{B}}\left[  X_{A}^{-1/2}\otimes I_{\hat{A}%
}\right]  X_{AB}^{1/2}\right)  \left(  \tau_{AB}^{-1}\otimes Y_{\hat{A}\hat
{B}}^{T}\right)  \left(  X_{AB}^{1/2}\left[  X_{A}^{-1/2}\otimes I_{\hat{A}%
}\right]  |\Gamma\rangle_{B\hat{B}}\right)  \\
&  =\langle\Gamma|_{B\hat{B}}\left(  X_{A}^{-1/2}X_{AB}^{1/2}\tau_{AB}%
^{-1}X_{AB}^{1/2}X_{A}^{-1/2}\otimes Y_{\hat{A}\hat{B}}^{T}\right)
|\Gamma\rangle_{B\hat{B}}\\
%&  =\langle\Gamma|_{B\hat{B}}\left(  \omega_{A}^{-1/2}\omega_{A}^{1/2}%
%X_{A}^{-1/2}X_{AB}^{1/2}\tau_{AB}^{-1}X_{AB}^{1/2}X_{A}^{-1/2}\omega_{A}%
%^{1/2}\omega_{A}^{-1/2}\otimes Y_{\hat{A}\hat{B}}^{T}\right)  |\Gamma
%\rangle_{B\hat{B}}\\
%&  =\langle\Gamma|_{B\hat{B}}\left(  \omega_{A}^{-1/2}\omega_{A}^{-1/2}\otimes
%Y_{\hat{A}\hat{B}}^{T}\right)  |\Gamma\rangle_{B\hat{B}}\\
&  =\langle\Gamma|_{B\hat{B}}\left(  \omega_{A}^{-1}\otimes Y_{\hat{A}\hat{B}%
}^{T}\right)  |\Gamma\rangle_{B\hat{B}}\\
&  =\omega_{A}^{-1}\otimes\langle\Gamma|_{B\hat{B}}Y_{\hat{A}\hat{B}}%
^{T}|\Gamma\rangle_{B\hat{B}}\\
&  =\omega_{A}^{-1}\otimes Y_{\hat{A}}^{T}.
\end{align}
For the fourth equality, we used the fact that $\tau_{AB}^{-1} = X_{AB}^{-1/2}(  [  X_{A}^{1/2}\omega_{A}^{-1} X_{A}%
^{1/2}]  \otimes I_{B})  X_{AB}^{-1/2}$ for the choice of $\tau_{AB}$ in \eqref{eq:tau-state-choice}.
With this setup, we can now readily establish the desired inequality by
employing the operator Jensen inequality \cite{HP03} and operator convexity of
the function $f$:%
\begin{align}
\widetilde{Q}_{f}(X_{AB}\Vert Y_{AB};\tau_{AB}) &  =\langle\varphi^{X_{AB}%
}|_{AB\hat{A}\hat{B}}f(\tau_{AB}^{-1}\otimes Y_{\hat{A}\hat{B}}^{T}%
)|\varphi^{X_{AB}}\rangle_{AB\hat{A}\hat{B}}\label{eq:op-convex-1}\\
&  =\langle\varphi^{X_{A}}|_{A\hat{A}}V^{\dag}f(\tau_{AB}^{-1}\otimes
Y_{\hat{A}\hat{B}}^{T})V|\varphi^{X_{A}}\rangle_{A\hat{A}}\\
&  \geq\langle\varphi^{X_{A}}|_{A\hat{A}}f(V^{\dag}[\tau_{AB}^{-1}\otimes
Y_{\hat{A}\hat{B}}^{T}]V)|\varphi^{X_{A}}\rangle_{A\hat{A}}\\
&  =\langle\varphi^{X_{A}}|_{A\hat{A}}f(\omega_{A}^{-1}\otimes Y_{\hat{A}}%
^{T})|\varphi^{X_{A}}\rangle_{A\hat{A}}\\
&  =\widetilde{Q}_{f}(X_{A}\Vert Y_{A};\omega_{A}).\label{eq:op-convex-last}%
\end{align}
Taking a supremum over $\tau_{AB}$ such that $\tau_{AB}>0$ and
$\operatorname{Tr}\{\tau_{AB}\}=1$, we conclude that the following inequality
holds for all invertible states $\omega_{A}$:%
\begin{equation}
\widetilde{Q}_{f}(X_{AB}\Vert Y_{AB})\geq\widetilde{Q}_{f}(X_{A}\Vert
Y_{A};\omega_{A}).
\end{equation}
After taking a supremum over invertible states $\omega_{A}$, we find that the
inequality in \eqref{eq:DP-partial-trace} holds when $X_{AB}$ is invertible.

\vspace{.1in}

\textbf{Case }$X_{A}>0$\textbf{, but }$X_{AB}$\textbf{ not invertible:} Consider the following isometry:%
\begin{equation}
V_{A\hat{A}\rightarrow AB\hat{A}\hat{B}}\equiv X_{AB}^{1/2}\left[
X_{A}^{-1/2}\otimes I_{\hat{A}}\right]  |\Gamma\rangle_{B\hat{B}}.
\end{equation}
The operator $V_{A\hat{A}\rightarrow A\hat{A}B\hat{B}}$ is indeed an isometry
 because
\begin{align}
V^{\dag}V &  =\left(  \langle\Gamma|_{B\hat{B}}\left[  X_{A}^{-1/2}\otimes
I_{\hat{A}}\right]  X_{AB}^{1/2}\right)  \left(  X_{AB}^{1/2}\left[
X_{A}^{-1/2}\otimes I_{\hat{A}}\right]  |\Gamma\rangle_{B\hat{B}}\right)  \\
&  =\langle\Gamma|_{B\hat{B}}\left[  X_{A}^{-1/2}\otimes I_{\hat{A}}\right]
X_{AB}\left[  X_{A}^{-1/2}\otimes I_{\hat{A}}\right]  |\Gamma\rangle_{B\hat
{B}}\\
&  =\left[  X_{A}^{-1/2}\otimes I_{\hat{A}}\right]  \langle\Gamma|_{B\hat{B}%
}X_{AB}|\Gamma\rangle_{B\hat{B}}\left[  X_{A}^{-1/2}\otimes I_{\hat{A}%
}\right]  \\
&  =\left[  X_{A}^{-1/2}\otimes I_{\hat{A}}\right]  X_{A}\left[  X_{A}%
^{-1/2}\otimes I_{\hat{A}}\right]  \\
&  =X_{A}^{-1/2}X_{A}X_{A}^{-1/2}\otimes I_{\hat{A}}\\
&  =I_{A}\otimes I_{\hat{A}}.
\end{align}
Then, for $\delta\in(0,1)$ and $\omega_{A}$ an invertible density operator,
take $\tau_{AB}$ to be the following invertible density operator:%
\begin{equation}
\tau_{AB}=\left(  1-\delta\right)  X_{AB}^{1/2}\left(  \left[  X_{A}%
^{-1/2}\omega_{A}X_{A}^{-1/2}\right]  \otimes I_{B}\right)  X_{AB}%
^{1/2}+\delta\xi_{AB},
\end{equation}
where $\xi_{AB}$ is some invertible density operator. We then find that%
\begin{align}
& V^{\dag}(\tau_{AB}^{-1}\otimes Y_{\hat{A}\hat{B}}^{T})V\nonumber\\
& =\langle\Gamma|_{B\hat{B}}X_{A}^{-1/2}X_{AB}^{1/2}(\tau_{AB}^{-1}\otimes
Y_{\hat{A}\hat{B}}^{T})X_{AB}^{1/2}X_{A}^{-1/2}|\Gamma\rangle_{B\hat{B}%
}\label{eq:op-collapse-1-new}\\
& =\langle\Gamma|_{B\hat{B}}X_{A}^{-1/2}X_{AB}^{1/2}\tau_{AB}^{-1}X_{AB}%
^{1/2}X_{A}^{-1/2}\otimes Y_{\hat{A}\hat{B}}^{T}|\Gamma\rangle_{B\hat{B}}\\
& =\langle\Gamma|_{B\hat{B}}\omega_{A}^{-1/2}\omega_{A}^{1/2}X_{A}%
^{-1/2}X_{AB}^{1/2}\tau_{AB}^{-1}X_{AB}^{1/2}X_{A}^{-1/2}\omega_{A}%
^{1/2}\omega_{A}^{-1/2}\otimes Y_{\hat{A}\hat{B}}^{T}|\Gamma\rangle_{B\hat{B}%
}\\
& \leq\langle\Gamma|_{B\hat{B}}\omega_{A}^{-1/2}\left\Vert \omega_{A}%
^{1/2}X_{A}^{-1/2}X_{AB}^{1/2}\tau_{AB}^{-1}X_{AB}^{1/2}X_{A}^{-1/2}\omega
_{A}^{1/2}\right\Vert _{\infty}\omega_{A}^{-1/2}\otimes Y_{\hat{A}\hat{B}}%
^{T}|\Gamma\rangle_{B\hat{B}}\\
& =\left\Vert \tau_{AB}^{-1/2}X_{AB}^{1/2}X_{A}^{-1/2}\omega_{A}X_{A}%
^{-1/2}X_{AB}^{1/2}\tau_{AB}^{-1/2}\right\Vert _{\infty}\langle\Gamma
|_{B\hat{B}}\omega_{A}^{-1}\otimes Y_{\hat{A}\hat{B}}^{T}|\Gamma\rangle
_{B\hat{B}}\\
& \leq\frac{1}{1-\delta}\left[  \omega_{A}^{-1}\otimes Y_{\hat{A}}^{T}\right]
.\label{eq:op-collapse-last-new}%
\end{align}
The last equality follows because $\Vert Z^\dag Z \Vert_\infty =  \Vert Z Z^\dag \Vert_\infty$ for any operator $Z$ (here we set $Z = \tau_{AB}^{-1/2}X_{AB}^{1/2}X_{A}^{-1/2}\omega
_{A}^{1/2}$.
The last inequality follows because%
\begin{align}
&  \left\Vert \tau_{AB}^{-1/2}X_{AB}^{1/2}X_{A}^{-1/2}\omega_{A}X_{A}%
^{-1/2}X_{AB}^{1/2}\tau_{AB}^{-1/2}\right\Vert _{\infty}\nonumber\\
&  =\inf\left\{  \mu:\tau_{AB}^{-1/2}X_{AB}^{1/2}X_{A}^{-1/2}\omega_{A}%
X_{A}^{-1/2}X_{AB}^{1/2}\tau_{AB}^{-1/2}\leq\mu I_{AB}\right\}  \\
&  =\inf\left\{  \mu:X_{AB}^{1/2}X_{A}^{-1/2}\omega_{A}X_{A}^{-1/2}%
X_{AB}^{1/2}\leq\mu\ \tau_{AB}\right\}  \\
&  =\inf\Bigg\{\mu:X_{AB}^{1/2}X_{A}^{-1/2}\omega_{A}X_{A}^{-1/2}X_{AB}%
^{1/2}\leq\mu\left[  \left(  1-\delta\right)  X_{AB}^{1/2}X_{A}^{-1/2}%
\omega_{A}X_{A}^{-1/2}X_{AB}^{1/2}+\delta\xi_{AB}\right]  \Bigg\}\\
&  \leq\frac{1}{1-\delta}.
\end{align}
Then consider that%
\begin{align}
\langle\varphi^{X_{AB}}|_{AB\hat{A}\hat{B}}f(\tau_{AB}^{-1}\otimes Y_{\hat
{A}\hat{B}}^{T})|\varphi^{X_{AB}}\rangle_{AB\hat{A}\hat{B}}  & =\langle
\varphi^{X_{A}}|_{A\hat{A}}V^{\dag}f(\tau_{AB}^{-1}\otimes Y_{\hat{A}\hat{B}%
}^{T})V|\varphi^{X_{A}}\rangle_{A\hat{A}}\\
& \geq\langle\varphi^{X_{A}}|_{A\hat{A}}f(V^{\dag}\left[  \tau_{AB}%
^{-1}\otimes Y_{\hat{A}\hat{B}}^{T}\right]  V)|\varphi^{X_{A}}\rangle
_{A\hat{A}}\\
& \geq\langle\varphi^{X_{A}}|_{A\hat{A}}f(\left[  1-\delta\right]
^{-1}\left[  \omega_{A}^{-1}\otimes Y_{\hat{A}}^{T}\right]  )|\varphi^{X_{A}%
}\rangle_{A\hat{A}},
\end{align}
where the first inequality is a consequence of the operator Jensen inequality
and the second follows from
\eqref{eq:op-collapse-1-new}--\eqref{eq:op-collapse-last-new} and operator
anti-monotonicity of the function $f$. Taking a supremum over invertible
density operators $\tau_{AB}$, we then conclude that the following inequality
holds for all $\delta\in(0,1)$ and for all invertible density operators
$\omega_{A}$:%
\begin{equation}
\widetilde{Q}_{f}(X_{AB}\Vert Y_{AB})\geq\langle\varphi^{X_{A}}|_{A\hat{A}%
}f(\left[  1-\delta\right]  ^{-1}\left[  \omega_{A}^{-1}\otimes Y_{\hat{A}%
}^{T}\right]  )|\varphi^{X_{A}}\rangle_{A\hat{A}}.
\end{equation}
Since this inequality holds for all $\delta\in(0,1)$ and for all invertible
density operators $\omega_{A}$, we can appeal to continuity of the function $f$ (taking the limit $\delta \searrow 0$) and then take a supremum over all invertible
density operators $\omega_{A}$ to
conclude the desired inequality for the case $X_A > 0 $, but $X_{AB}$ is not invertible:%
\begin{equation}
\widetilde{Q}_{f}(X_{AB}\Vert Y_{AB})\geq\widetilde{Q}_{f}(X_{A}\Vert Y_{A}),
\end{equation}
as claimed.
\end{proof}

\begin{remark}
\label{rem:op-convex-inv}I stress once again here that if $X_{AB}$ and
$Y_{AB}$ are invertible, then we only require operator convexity of the
function $f$ in order to arrive at the inequality in
\eqref{eq:DP-partial-trace}. One can examine the steps in
\eqref{eq:op-convex-1}--\eqref{eq:op-convex-last} to see this.
\end{remark}

Based on Propositions~\ref{prop:iso-inv}\ and \ref{prop:mono-PT}\ and the
Stinespring dilation theorem \cite{S55}, we conclude the following
data-processing theorem for the optimized quantum $f$-divergences:

\begin{theorem}
[Quantum data processing]\label{thm:DP}Let $X_{S}$ and $Y_{S}$ be positive
semi-definite operators acting on a Hilbert space $\mathcal{H}_{S}$, and let
$\mathcal{N}_{S\rightarrow B}$ be a quantum channel taking operators acting on
$\mathcal{H}_{S}$ to operators acting on a Hilbert space $\mathcal{H}_{B}$.
Let $f$ be an operator anti-monotone function with domain $(0,\infty)$ and
range $\mathbb{R}$. Then the following inequality holds%
\begin{equation}
\widetilde{Q}_{f}(X_{S}\Vert Y_{S})\geq\widetilde{Q}_{f}(\mathcal{N}%
_{S\rightarrow B}(X_{S})\Vert\mathcal{N}_{S\rightarrow B}(Y_{S})).
\end{equation}

\end{theorem}

\section{Examples of optimized quantum $f$-divergences}

\label{sec:examples-opt-f-div}I now show how several known quantum divergences
are particular examples of an optimized quantum $f$-divergence, including the
quantum relative entropy \cite{U62}\ and the sandwiched R\'{e}nyi relative
quasi-entropies \cite{MDSFT13,WWY13}. The result will be that
Theorem~\ref{thm:DP} recovers quantum data processing for the sandwiched
R\'{e}nyi relative entropies for the full range of parameters for which it is
known to hold. Thus, one benefit of Theorem~\ref{thm:DP} and earlier work of
\cite{P85,P86,TCR09}\ is a single, unified approach, based on the operator
Jensen inequality \cite{HP03}, for establishing quantum data processing for
all of the Petz-- and sandwiched R\'{e}nyi relative entropies for the full
parameter ranges for which data processing is known to hold.

\subsection{Quantum relative entropy as optimized quantum $f$-divergence}

\label{sec:q-rel-ent}Let $\tau$ be an invertible state, $\varepsilon>0$, and
$X$ and $Y$ positive semi-definite. Let $\overline{X}=X/\operatorname{Tr}%
\{X\}$. Pick the function%
\begin{equation}
f(x)=-\log x,
\end{equation}
which is an operator anti-monotone function with domain $(0,\infty)$ and range
$\mathbb{R}$, and we find that%
\begin{align}
&  \frac{1}{\operatorname{Tr}\{X\}}\langle\varphi^{X}|_{S\hat{S}}\left[
-\log(\tau_{S}^{-1}\otimes\left(  Y+\varepsilon\Pi_{Y}^{\perp}\right)
_{\hat{S}}^{T})\right]  |\varphi^{X}\rangle_{S\hat{S}}\nonumber\\
&  =\langle\varphi^{\overline{X}}|_{S\hat{S}}\left[  \log(\tau_{S})\otimes
I_{\hat{S}}-I_{S}\otimes\log\left(  Y+\varepsilon\Pi_{Y}^{\perp}\right)
_{\hat{S}}^{T}\right]  |\varphi^{\overline{X}}\rangle_{S\hat{S}}\\
&  =\langle\varphi^{\overline{X}}|_{S\hat{S}}\log(\tau_{S})\otimes I_{\hat{S}%
}|\varphi^{\overline{X}}\rangle_{S\hat{S}}-\langle\varphi^{\overline{X}%
}|_{S\hat{S}}I_{S}\otimes\log\left(  Y+\varepsilon\Pi_{Y}^{\perp}\right)
_{\hat{S}}^{T}|\varphi^{\overline{X}}\rangle_{S\hat{S}}\\
&  =\operatorname{Tr}\{\overline{X}\log\tau\}-\operatorname{Tr}\{\overline
{X}\log(Y+\varepsilon\Pi_{Y}^{\perp})\}\\
&  \leq\operatorname{Tr}\{\overline{X}\log\overline{X}\}-\operatorname{Tr}%
\{\overline{X}\log(Y+\varepsilon\Pi_{Y}^{\perp})\}\\
&  =D(\overline{X}\Vert Y+\varepsilon\Pi_{Y}^{\perp}).
\end{align}
The inequality is a consequence of Klein's inequality \cite{O31} (see also
\cite{R02}), establishing that the optimal $\tau$ is set to $\overline{X}%
$.\footnote{Technically, we would require an invertible $\tau$ that
approximates $\overline{X}$ arbitrarily well in order to achieve equality in
Klein's inequality. One can alternatively establish the inequality
$\operatorname{Tr}\{\overline{X}\log\overline{X}\}\geq\operatorname{Tr}%
\{\overline{X}\log\tau\}$ by employing the non-negativity of quantum relative
entropy $D(\overline{X}\Vert\tau)\geq0$ for quantum states.} Now taking a
supremum over $\varepsilon>0$, we find that%
\begin{equation}
\widetilde{Q}_{-\log(\cdot)}(X\Vert Y)=\operatorname{Tr}\{X\}D(\overline
{X}\Vert Y),
\end{equation}
where the quantum relative entropy $D(\overline{X}\Vert Y)$ is defined as
\cite{U62}%
\begin{equation}
D(\overline{X}\Vert Y)=\operatorname{Tr}\{\overline{X}\left[  \log\overline
{X}-\log Y\right]  \}
\end{equation}
if $\operatorname{supp}(X)\subseteq\operatorname{supp}(Y)$ and $D(\overline
{X}\Vert Y)=+\infty$ otherwise.

\subsection{Sandwiched R\'enyi relative quasi-entropy as optimized quantum
$f$-divergence}

Take $\tau$, $\varepsilon$, $X$, and $Y$ as defined in
Section~\ref{sec:q-rel-ent}. For $\alpha\in\lbrack1/2,1)$, pick the function%
\begin{equation}
f(x)=-x^{\left(  1-\alpha\right)  /\alpha},
\end{equation}
which is an operator anti-monotone function with domain $(0,\infty)$ and range
$\mathbb{R}$. Note that this is a reparametrization of $-x^{\beta}$ for
$\beta\in(0,1]$. I now show that%
\begin{equation}
\widetilde{Q}_{-\left(  \cdot\right)  ^{\left(  1-\alpha\right)  /\alpha}%
}(X\Vert Y)=-\left\Vert Y^{\left(  1-\alpha\right)  /2\alpha}XY^{\left(
1-\alpha\right)  /2\alpha}\right\Vert _{\alpha},
\end{equation}
which is the known expression for sandwiched R\'{e}nyi relative quasi-entropy
for $\alpha\in\lbrack1/2,1)$ \cite{MDSFT13,WWY13}. To see this, consider that%
\begin{align}
&  -\langle\varphi^{X}|_{S\hat{S}}\left[  \tau_{S}^{-1}\otimes\left(
Y+\varepsilon\Pi_{Y}^{\perp}\right)  _{\hat{S}}^{T}\right]  ^{\left(
1-\alpha\right)  /\alpha}|\varphi^{X}\rangle_{S\hat{S}}\nonumber\\
&  =-\langle\varphi^{X}|_{S\hat{S}}\tau_{S}^{\left(  \alpha-1\right)  /\alpha
}\otimes\left(  \left(  Y+\varepsilon\Pi_{Y}^{\perp}\right)  _{\hat{S}}%
^{T}\right)  ^{\left(  1-\alpha\right)  /\alpha}|\varphi^{X}\rangle_{S\hat{S}%
}\\
&  =-\langle\Gamma|_{S\hat{S}}X_{S}^{1/2}\tau_{S}^{\left(  \alpha-1\right)
/\alpha}X_{S}^{1/2}\otimes\left(  \left(  Y+\varepsilon\Pi_{Y}^{\perp}\right)
_{\hat{S}}^{T}\right)  ^{\left(  1-\alpha\right)  /\alpha}|\Gamma
\rangle_{S\hat{S}}\\
&  =-\operatorname{Tr}\left\{  X^{1/2}\tau^{\left(  \alpha-1\right)  /\alpha
}X^{1/2}(Y+\varepsilon\Pi_{Y}^{\perp})^{\left(  1-\alpha\right)  /\alpha
}\right\} \\
&  =-\operatorname{Tr}\left\{  X^{1/2}(Y+\varepsilon\Pi_{Y}^{\perp})^{\left(
1-\alpha\right)  /\alpha}X^{1/2}\tau^{\left(  \alpha-1\right)  /\alpha
}\right\}  .
\end{align}
The formulas in the above development are related to those given in the proof
of \cite[Lemma~19]{MDSFT13}. Now optimizing over invertible states $\tau$ and
employing H\"{o}lder duality \cite{B97}, in the form of the reverse H\"{o}lder inequality
and as observed in \cite{MDSFT13}, we find that%
\begin{equation}
\sup_{\substack{\tau>0,\\\operatorname{Tr}\{\tau\}=1}}\left[
-\operatorname{Tr}\left\{  X^{1/2}(Y+\varepsilon\Pi_{Y}^{\perp})^{\left(
1-\alpha\right)  /\alpha}X^{1/2}\tau^{\left(  \alpha-1\right)  /\alpha
}\right\}  \right]  =-\left\Vert X^{1/2}(Y+\varepsilon\Pi_{Y}^{\perp
})^{\left(  1-\alpha\right)  /\alpha}X^{1/2}\right\Vert _{\alpha},
\end{equation}
where for positive semi-definite $Z$, we define%
\begin{equation}
\left\Vert Z\right\Vert _{\alpha}=\left[  \operatorname{Tr}\{Z^{\alpha
}\}\right]  ^{1/\alpha}.
\end{equation}
Now taking the limit $\varepsilon\searrow0$, we get that%
\begin{equation}
\widetilde{Q}_{-\left(  \cdot\right)  ^{\left(  1-\alpha\right)  /\alpha}%
}(X\Vert Y)=-\left\Vert X^{1/2}Y^{\left(  1-\alpha\right)  /\alpha}%
X^{1/2}\right\Vert _{\alpha}=-\left\Vert Y^{\left(  1-\alpha\right)  /2\alpha
}XY^{\left(  1-\alpha\right)  /2\alpha}\right\Vert _{\alpha},
\end{equation}
which is the sandwiched R\'{e}nyi relative quasi-entropy for the range
$\alpha\in\lbrack1/2,1)$. The sandwiched R\'{e}nyi relative entropy itself is
defined up to a normalization factor as \cite{MDSFT13,WWY13}%
\begin{equation}
\widetilde{D}_{\alpha}(X\Vert Y)=\frac{\alpha}{\alpha-1}\log\left\Vert
Y^{\left(  1-\alpha\right)  /2\alpha}XY^{\left(  1-\alpha\right)  /2\alpha
}\right\Vert _{\alpha}. \label{eq:sandwiched-Renyi}%
\end{equation}
Thus, Theorem~\ref{thm:DP}\ implies quantum data processing for the sandwiched
R\'{e}nyi relative entropy%
\begin{equation}
\widetilde{D}_{\alpha}(X\Vert Y)\geq\widetilde{D}_{\alpha}(\mathcal{N}%
(X)\Vert\mathcal{N}(Y)),
\end{equation}
for the parameter range $\alpha\in\lbrack1/2,1)$, which is a result previously
established in \cite{FL13}.

For $\alpha\in(1,\infty]$, pick the function%
\begin{equation}
f(x)=x^{\left(  1-\alpha\right)  /\alpha},
\end{equation}
which is an operator anti-monotone function with domain $(0,\infty)$ and range
$\mathbb{R}$. Note that this is a reparametrization of $x^{\beta}$ for
$\beta\in\lbrack-1,0)$. I now show that%
\begin{equation}
\widetilde{Q}_{\left(  \cdot\right)  ^{\left(  1-\alpha\right)  /\alpha}%
}(X\Vert Y)=\left\{
\begin{array}
[c]{cc}%
\left\Vert Y^{\left(  1-\alpha\right)  /2\alpha}XY^{\left(  1-\alpha\right)
/2\alpha}\right\Vert _{\alpha} & \text{if }\operatorname{supp}(X)\subseteq
\operatorname{supp}(Y)\\
+\infty & \text{else}%
\end{array}
\right.  ,
\end{equation}
which is the known expression for sandwiched R\'{e}nyi relative quasi-entropy
for $\alpha\in(1,\infty]$ \cite{MDSFT13,WWY13}. To see this, consider that the
same development as above gives that%
\begin{equation}
\langle\varphi^{X}|_{S\hat{S}}(\tau_{S}^{-1}\otimes\left(  Y+\varepsilon
\Pi_{Y}^{\perp}\right)  _{\hat{S}}^{T})^{\left(  1-\alpha\right)  /\alpha
}|\varphi^{X}\rangle_{S\hat{S}}=\operatorname{Tr}\left\{  X^{1/2}%
(Y+\varepsilon\Pi_{Y}^{\perp})^{\left(  1-\alpha\right)  /\alpha}X^{1/2}%
\tau^{\left(  \alpha-1\right)  /\alpha}\right\}  .
\end{equation}
Again employing H\"{o}lder duality, as observed in \cite{MDSFT13}, we find
that%
\begin{equation}
\sup_{\tau>0,\operatorname{Tr}\{\tau\}=1}\operatorname{Tr}\left\{
X^{1/2}(Y+\varepsilon\Pi_{Y}^{\perp})^{\left(  1-\alpha\right)  /\alpha
}X^{1/2}\tau^{\left(  \alpha-1\right)  /\alpha}\right\}  =\left\Vert
X^{1/2}(Y+\varepsilon\Pi_{Y}^{\perp})^{\left(  1-\alpha\right)  /\alpha
}X^{1/2}\right\Vert _{\alpha},
\end{equation}
Now taking the limit $\varepsilon\searrow0$, we get that%
\begin{equation}
\widetilde{Q}_{\left(  \cdot\right)  ^{\left(  1-\alpha\right)  /\alpha}%
}(X\Vert Y)=\left\Vert X^{1/2}Y^{\left(  1-\alpha\right)  /\alpha}%
X^{1/2}\right\Vert _{\alpha}=\left\Vert Y^{\left(  1-\alpha\right)  /2\alpha
}XY^{\left(  1-\alpha\right)  /2\alpha}\right\Vert _{\alpha},
\end{equation}
where the equalities hold if $\operatorname{supp}(X)\subseteq
\operatorname{supp}(Y)$ and otherwise $\widetilde{Q}_{\left(  \cdot\right)
^{\left(  1-\alpha\right)  /\alpha}}(X\Vert Y)=+\infty$, as observed in
\cite{MDSFT13}. The sandwiched R\'{e}nyi relative entropy itself is defined up
to a normalization factor as in \eqref{eq:sandwiched-Renyi} if
$\operatorname{supp}(X)\subseteq\operatorname{supp}(Y)$ and otherwise
$\widetilde{D}_{\alpha}(X\Vert Y)=+\infty$ for $\alpha\in(1,\infty]$
\cite{MDSFT13,WWY13}. Thus, Theorem~\ref{thm:DP}\ implies quantum data
processing for the sandwiched R\'{e}nyi relative entropy%
\begin{equation}
\widetilde{D}_{\alpha}(X\Vert Y)\geq\widetilde{D}_{\alpha}(\mathcal{N}%
(X)\Vert\mathcal{N}(Y)),
\end{equation}
for the parameter range $\alpha\in(1,\infty]$, which is a result previously
established in full by \cite{FL13,B13monotone,MO13} and for $\alpha\in(1,2]$
by \cite{MDSFT13,WWY13}.

\subsection{Optimized $\alpha$-divergence: monotonicity under partial trace
for invertible density operators}

Interestingly, for $\alpha\in\lbrack1/3,1/2]$, the function%
\begin{equation}
f(x)=x^{(1-\alpha)/\alpha}%
\end{equation}
is operator convex on the domain $(0,\infty)$ and with range $\mathbb{R}$.
Note that this is a reparametrization of $x^{\beta}$ for $\beta\in\lbrack
1,2]$. Thus, by following the same development as before, for positive
definite $X$ and $Y$ we find that%
\begin{equation}
\langle\varphi^{X}|_{S\hat{S}}(\tau_{S}^{-1}\otimes Y_{\hat{S}}^{T})^{\beta
}|\varphi^{X}\rangle_{S\hat{S}}=\operatorname{Tr}\left\{  X^{1/2}Y^{\beta
}X^{1/2}\tau^{-\beta}\right\}  .
\end{equation}
Now optimizing over $\tau$, we find that the following function%
\begin{equation}
\widetilde{Q}_{(\cdot)^{\beta}}(X\Vert Y)=\sup_{\tau>0,\operatorname{Tr}%
\{\tau\}=1}\operatorname{Tr}\left\{  X^{1/2}Y^{\beta}X^{1/2}\tau^{-\beta
}\right\}
\end{equation}
is monotone with respect to partial trace for $\beta\in\lbrack1,2]$. That is,
the inequality%
\begin{equation}
\widetilde{Q}_{(\cdot)^{\beta}}(X_{AB}\Vert Y_{AB})\geq\widetilde{Q}%
_{(\cdot)^{\beta}}(X_{A}\Vert Y_{A})
\end{equation}
holds for $\beta\in\lbrack1,2]$ and positive definite $X_{AB}$ and $Y_{AB}$,
by applying Remark~\ref{rem:op-convex-inv}.

Take note that $\widetilde{Q}_{(\cdot)^{\beta}}(X\Vert Y)$ for $\beta
\in\lbrack1,2]$ is not a sandwiched R\'enyi relative quasi-entropy because the
optimization over $\tau$ goes the opposite way when compared to that for the
sandwiched R\'enyi relative entropy for $\alpha\in\lbrack1/3,1/2]$. This is
consistent with the fact that data processing is known not to hold for the
sandwiched R\'enyi relative entropy for $\alpha\in(0,1/2)$
\cite{MDSFT13,DL14,BFT15}.

\section{On Petz's quantum $f$-divergence}

\label{sec:petz-f}I now discuss in more detail the relation between the
optimized quantum $f$-divergence and the Petz quantum $f$-divergence from
\cite{P85,P86}. In brief, we find that the Petz $f$-divergence can be
recovered by replacing $\tau$ in Definition~\ref{def:opt-f-div} with
$X+\delta\pi_{X}^{\perp}$.

\begin{definition}
[Petz quantum $f$-divergence]\label{def:petz-f-div}Let $f$ be a continuous
function with domain $(0,\infty)$ and range $\mathbb{R}$. For positive
semi-definite operators $X$ and $Y$ acting on a Hilbert space $\mathcal{H}%
_{S}$, the Petz quantum $f$-divergence is defined as%
\begin{equation}
Q_{f}(X\Vert Y)\equiv\sup_{\varepsilon>0}\lim_{\delta\searrow0}\langle
\varphi^{X}|_{S\hat{S}}f\!\left(  \left[  X_{S}+\delta\pi_{X}^{\perp}\right]
^{-1}\otimes\left[  Y_{\hat{S}}+\varepsilon\Pi_{Y}^{\perp}\right]
^{T}\right)  |\varphi^{X}\rangle_{S\hat{S}},
\end{equation}
where $\pi_{X}^{\perp}=\Pi_{X}^{\perp}/\operatorname{Tr}\{\Pi_{X}^{\perp}\}$
is the maximally mixed state on the kernel of $X$ and the rest of the notation
is the same as in Definition$~$\ref{def:opt-f-div}. If the kernel of $X$ is
equal to zero, then we set $\pi_{X}^{\perp}=0$.
\end{definition}

Let spectral decompositions of positive semi-definite $X$ and positive
definite $Y$ be given as%
\begin{equation}
X=\sum_{x}\lambda_{x}|\psi^{x}\rangle\langle\psi^{x}|,\qquad Y=\sum_{y}\mu
_{y}|\phi^{y}\rangle\langle\phi^{y}|.
\end{equation}
By following the same development needed to arrive at
\eqref{eq:f-divergence-spec-form}, we see that $Q_{f}(X\Vert Y)$ can be
written for non-invertible $X$ and invertible $Y$ as%
\begin{align}
Q_{f}(X\Vert Y)  &  =\lim_{\delta\searrow0}\sum_{y}\Bigg[\sum_{x:\lambda
_{x}\neq0}f(\mu_{y}\lambda_{x}^{-1})\operatorname{Tr}\{X^{1/2}|\psi^{x}%
\rangle\langle\psi^{x}|X^{1/2}|\phi^{y}\rangle\langle\phi^{y}|\}\nonumber\\
&  \qquad\qquad+f(\mu_{y}\operatorname{Tr}\{\Pi_{X}^{\perp}\}\delta
^{-1})\operatorname{Tr}\{X^{1/2}\Pi_{X}^{\perp}X^{1/2}|\phi^{y}\rangle
\langle\phi^{y}|\}\Bigg]\\
&  =\lim_{\delta\searrow0}\sum_{y}\sum_{x:\lambda_{x}\neq0}f(\mu_{y}%
\lambda_{x}^{-1})\operatorname{Tr}\{X^{1/2}|\psi^{x}\rangle\langle\psi
^{x}|X^{1/2}|\phi^{y}\rangle\langle\phi^{y}|\}\\
&  =\sum_{y}\sum_{x:\lambda_{x}\neq0}\lambda_{x}f(\mu_{y}\lambda_{x}%
^{-1})\left\vert \langle\psi^{x}|\phi^{y}\rangle\right\vert ^{2}.
\end{align}
Note that we get the same formula for $Q_{f}(X\Vert Y)$ if $X$ is invertible.
For non-invertible $Y$, we just substitute $Y+\varepsilon\Pi_{Y}^{\perp}$ and
take the supremum over $\varepsilon>0$ at the end.

The next concern is about quantum data processing with the Petz $f$-divergence
as defined above. To show this, we take $f$ to be an operator anti-monotone
function with domain $(0,\infty)$ and range $\mathbb{R}$. As discussed in
Section~\ref{sec:DP}, one can establish data processing by showing isometric
invariance and monotonicity under partial trace. Isometric invariance of
$Q_{f}(X\Vert Y)$ follows from the same proof as given in
Proposition~\ref{prop:iso-inv}. Monotonicity of $Q_{f}(X_{AB}\Vert Y_{AB})$
under partial trace breaks down into three cases depending on invertibility of
$X_{AB}$ or $X_{A}$, as discussed in the proof of
Proposition~\ref{prop:mono-PT}. For the proof, we assume as previously done
that $Y_{AB}$ is invertible throughout. If it is not, then
Definition~\ref{def:petz-f-div}\ forces it to be invertible and then a
supremum over $\varepsilon>0$ is finally taken at the end.

\begin{enumerate}
\item The case when $X_{AB}$ is invertible is already handled by Petz's proof
from \cite{P85,P86}, which relies on the operator Jensen inequality
\cite{HP03}. In this case, the operator $X_{AB}+\delta\pi_{AB}^{\perp}$
reduces to $X_{AB}$ because $\Pi_{AB}^{\perp}=0$.

\item The case when $X_{AB}$ is not invertible but $X_{A}$ is can be
understood as an appeal to continuity, as discussed in
Remark~\ref{rem:cont-appeal}. For this case, we take the operator $\tau_{AB}$
%in
%\eqref{eq:petz-channel-final-ext}
for some $\delta_{1}\in(0,1)$ to be %
%\begin{equation}
%Z_{A}\rightarrow\left(  1-\delta_{1}\right)  X_{AB}^{1/2}\left(  \left[
%X_{A}^{-1/2}Z_{A}X_{A}^{-1/2}\right]  \otimes I_{B}\right)  X_{AB}%
%^{1/2}+\delta_{1}\pi_{X_{AB}}^{\perp}.
%\end{equation}
%Inputting $X_{A}$ then leads to the output
$\left(  1-\delta_{1}\right)
X_{AB}+\delta_{1}\pi_{X_{AB}}^{\perp}$, which is a positive definite operator.
The rest of the proof proceeds the same and then the monotonicity under
partial trace holds for this case.

\item As far as I can tell, the case when $X_{A}$ is not invertible was not
discussed in several papers of Petz \textit{et al}%
.~\cite{P85,P86,PS09,P10,P10a}, and it was only considered recently in
\cite[Proposition~3.12]{HM17}. However, the method I have given in the proof
of Proposition~\ref{prop:mono-PT}\ appears to be different. Also,
Remark~\ref{rem:cont-appeal} in the appendix discusses how this approach
arguably extends beyond a mere appeal to continuity. For this case, we take
the channel in \eqref{eq:extended-petz-channel}\ to be%
\begin{equation}
Z_{A}\rightarrow X_{AB}^{1/2}\left(  \left[  X_{A}^{-1/2}Z_{A}X_{A}%
^{-1/2}\right]  \otimes I_{B}\right)  X_{AB}^{1/2}+\operatorname{Tr}%
\{\Pi_{X_{A}}^{\perp}Z_{A}\}\pi_{X_{AB}}^{\perp}.
\end{equation}
Inputting $X_{A}+\delta\pi_{A}^{\perp}$ then leads to the output
$X_{AB}+\delta\pi_{AB}^{\perp}$, which is a positive definite operator. The
rest of the proof proceeds the same and then the monotonicity under partial
trace holds for this case.
\end{enumerate}

Special and interesting cases of the Petz $f$-divergence are found by taking%
\begin{align}
f(x)  &  =-\log x,\\
f(x)  &  =-x^{\beta}\text{\qquad for }\beta\in(0,1],\\
f(x)  &  =x^{\beta}\text{\qquad for }\beta\in\lbrack-1,0).
\end{align}
Each of these functions are operator anti-monotone with domain $(0,\infty)$
and range $\mathbb{R}$. By following similar reasoning as in
Section~\ref{sec:examples-opt-f-div} to simplify $Q_{f}$ and by applying the
above arguments for data processing, we find that all of the following
quantities obey the data processing inequality:%
\begin{align}
Q_{-\log(\cdot)}(X\Vert Y)  &  =\operatorname{Tr}\{X\}D(\overline{X}\Vert
Y),\\
Q_{-(\cdot)^{\beta}}(X\Vert Y)  &  =-\operatorname{Tr}\{X^{1-\beta}Y^{\beta
}\},\text{ for }\beta\in(0,1],\\
Q_{(\cdot)^{\beta}}(X\Vert Y)  &  =\left\{
\begin{array}
[c]{cc}%
\operatorname{Tr}\{X^{1-\beta}Y^{\beta}\} & \text{if\ }\operatorname{supp}%
(X)\subseteq\operatorname{supp}(Y)\\
+\infty & \text{else}%
\end{array}
\right.  ,\ \text{for }\beta\in\lbrack-1,0),
\end{align}
where again $\overline{X}=X/\operatorname{Tr}\{X\}$. By a reparametrization
$\alpha=1-\beta$, we find that the latter two quantities are directly related
to the Petz R\'enyi relative entropy, defined as%
\begin{equation}
D_{\alpha}(X\Vert Y)\equiv\left\{
\begin{array}
[c]{cc}%
\frac{1}{\alpha-1}\log\operatorname{Tr}\{X^{\alpha}Y^{1-\alpha}\} &
\text{if\ }\operatorname{supp}(X)\subseteq\operatorname{supp}(Y)\text{ and
}\alpha>1\\
+\infty & \text{else}%
\end{array}
\right.  .
\end{equation}
Thus, the data processing proof establishes the data processing inequality for
$D_{\alpha}(X\Vert Y)$ for $\alpha\in\lbrack0,1)\cup(1,2]$, which is the range
for which it was already known to hold from prior work \cite{P86,TCR09}.

\begin{remark}
One beneficial aspect of the present paper is that we now see that there is a
single, unified approach, based on the operator Jensen inequality, for
establishing the data processing inequality for both the Petz--R\'enyi
relative entropy for $\alpha\in\lbrack0,1)\cup(1,2]$ and the sandwiched
R\'enyi relative entropy for $\alpha\in\lbrack1/2,1)\cup(1,\infty]$, the full
ranges of $\alpha$ for which the data processing inequality is already known
from \cite{P85,P86,TCR09,MDSFT13,WWY13,FL13,B13monotone,MO13} to hold for
these quantities. Prior to the present paper, there were a variety of
different ways for establishing the data processing inequality for the
sandwiched R\'enyi relative entropy, which can be found in
\cite{MDSFT13,WWY13,FL13,B13monotone,MO13}.
\end{remark}

Interestingly, for $\beta\in\lbrack1,2]$, the function%
\begin{equation}
f(x)=x^{\beta}%
\end{equation}
is operator convex on the domain $(0,\infty)$ and with range $\mathbb{R}$.
Thus, for positive definite $X$ and $Y$ we find that%
\begin{equation}
Q_{(\cdot)^{\beta}}(X\Vert Y)=\operatorname{Tr}\left\{  X^{1-\beta}Y^{\beta
}\right\}
\end{equation}
is monotone with respect to partial trace for $\beta\in\lbrack1,2]$. That is,
the inequality%
\begin{equation}
Q_{(\cdot)^{\beta}}(X_{AB}\Vert Y_{AB})\geq Q_{(\cdot)^{\beta}}(X_{A}\Vert
Y_{A})
\end{equation}
holds for $\beta\in\lbrack1,2]$ and positive definite $X_{AB}$ and $Y_{AB}$,
by applying Remark~\ref{rem:op-convex-inv}. By reparametrizing with
$\alpha=1-\beta$, we find that the following inequality holds for positive
definite $X_{AB}$ and $Y_{AB}$ and $\alpha\in\lbrack-1,0]$:%
\begin{equation}
D_{\alpha}(X_{AB}\Vert Y_{AB})\leq D_{\alpha}(X_{A}\Vert Y_{A}).
\end{equation}
Note that there is trivially an equality when $\alpha=0$, under the assumption
that $X_{AB}$ and $Y_{AB}$ are positive definite, because%
\begin{equation}
D_{0}(X_{AB}\Vert Y_{AB})=-\log\operatorname{Tr}\{Y_{AB}\}=-\log
\operatorname{Tr}\{Y_{A}\}=D_{0}(X_{A}\Vert Y_{A}).
\end{equation}

\subsection{Inequality for sandwiched and Petz--R\'{e}nyi relative entropies}

The development above motivates the following inequality relating the
sandwiched and Petz--R\'{e}nyi relative entropies. The same inequality was
shown in \cite{J18} when $X$ and $Y$ are normal states of an arbitrary von
Neumann algebra and for $\alpha> 1$, whereas the following proposition
considers the case when $X$ and $Y$ are positive semi-definite operators
acting on a finite-dimensional Hilbert space and the range $\alpha\in\lbrack
1/2,1)\cup(1,\infty)$.

\begin{proposition}
Let $X$ and $Y$ be positive semi-definite operators such that $X,Y\neq0$. Then
the following inequality holds for $\alpha\in\lbrack1/2,1)\cup(1,\infty)$:%
\begin{equation}
\widetilde{D}_{\alpha}(X\Vert Y)\geq D_{\left(  2\alpha-1\right)  /\alpha
}(X\Vert Y)-\log\operatorname{Tr}\{X\}.
\end{equation}

\end{proposition}

\begin{proof}
Without loss of generality, let us assume that $X$ and $Y$ are invertible. The
above inequality follows simply by picking the state $\tau=X/\operatorname{Tr}%
\{X\}$. Indeed, let $\alpha\in\lbrack1/2,1)$. Consider that%
\begin{align}
-\left\Vert X^{1/2}Y^{\left(  1-\alpha\right)  /\alpha}X^{1/2}\right\Vert
_{\alpha}  &  =\sup_{\substack{\tau>0,\\\operatorname{Tr}\{\tau\}=1}}\left[
-\operatorname{Tr}\left\{  X^{1/2}Y^{\left(  1-\alpha\right)  /\alpha}%
X^{1/2}\tau^{\left(  \alpha-1\right)  /\alpha}\right\}  \right] \\
&  \geq-\operatorname{Tr}\left\{  X^{1/2}Y^{\left(  1-\alpha\right)  /\alpha
}X^{1/2}\left(  X/\operatorname{Tr}\{X\}\right)  ^{\left(  \alpha-1\right)
/\alpha}\right\} \\
&  =-\operatorname{Tr}\left\{  Y^{\left(  1-\alpha\right)  /\alpha}X^{\left(
2\alpha-1\right)  /\alpha}\right\}  \left[  \operatorname{Tr}\{X\}\right]
^{\left(  1-\alpha\right)  /\alpha}.
\end{align}
This inequality implies that%
\begin{equation}
\log\left\Vert X^{1/2}Y^{\left(  1-\alpha\right)  /\alpha}X^{1/2}\right\Vert
_{\alpha}\leq\log\operatorname{Tr}\left\{  X^{\left(  2\alpha-1\right)
/\alpha}Y^{\left(  1-\alpha\right)  /\alpha}\right\}  +\frac{1-\alpha}{\alpha
}\log\operatorname{Tr}\{X\}.
\end{equation}
Multiplying by $\frac{\alpha}{\alpha-1}$ leads to%
\begin{align}
\widetilde{D}_{\alpha}(X\Vert Y)  &  =\frac{\alpha}{\alpha-1}\log\left\Vert
X^{1/2}Y^{\left(  1-\alpha\right)  /\alpha}X^{1/2}\right\Vert _{\alpha}\\
&  \geq\frac{\alpha}{\alpha-1}\log\operatorname{Tr}\left\{  X^{\left(
2\alpha-1\right)  /\alpha}Y^{\left(  1-\alpha\right)  /\alpha}\right\}
-\log\operatorname{Tr}\{X\}\\
&  =\frac{1}{\left[  \frac{2\alpha-1}{\alpha}\right]  -1}\log\operatorname{Tr}%
\left\{  X^{\left(  2\alpha-1\right)  /\alpha}Y^{\left(  1-\alpha\right)
/\alpha}\right\}  -\log\operatorname{Tr}\{X\}\\
&  =D_{\frac{2\alpha-1}{\alpha}}(X\Vert Y)-\log\operatorname{Tr}\{X\}.
\end{align}
This establishes the claim for $\alpha\in\lbrack1/2,1)$. The proof for
$\alpha\in(1,\infty)$ is very similar.
\end{proof}

\section{Classical and classical--quantum cases}

\label{sec:c-cq}When the operators $X$ and $Y$ commute, the optimized
$f$-divergence takes on a simpler form, as stated in the following proposition:

\begin{proposition}
[Classical case]\label{prop:classical-case}Let $f$ be an operator
anti-monotone function with domain $(0,\infty)$ and range $\mathbb{R}$. Let
$X$ and $Y$ be positive semi-definite operators that commute, having spectral
decompositions%
\begin{equation}
X=\sum_{z}\lambda_{z}|z\rangle\langle z|,\qquad Y=\sum_{z}\mu_{z}%
|z\rangle\langle z|,
\end{equation}
for a common eigenbasis $\{|z\rangle\}_{z}$. Then%
\begin{equation}
\widetilde{Q}_{f}(X\Vert Y)=\sup_{\left\{  \tau_{z}\right\}  _{z}%
,\ \varepsilon>0}\left\{  \sum_{z\, :\, \mu_{z}\neq0}\lambda_{z}f(\mu_{z}%
/\tau_{z})+\sum_{z\, :\, \mu_{z}=0}\lambda_{z}f(\varepsilon/\tau_{z}):\tau
_{z}>0\ \forall z,\ \sum_{z}\tau_{z}=1\right\}  .
\end{equation}

\end{proposition}

\begin{proof}
For simplicity, we prove the statement for the case in which $Y$ is
invertible, and then the extension to non-invertible $Y$ is straightforward.
For a spectral decomposition of $\tau$ as $\tau=\sum_{t}\nu_{t}|\phi
_{t}\rangle\langle\phi_{t}|$ and by applying
\eqref{eq:f-divergence-spec-form}, we find that%
\begin{align}
\widetilde{Q}_{f}(X\Vert Y;\tau)  &  =\sum_{z,t}f(\mu_{z}/\nu_{t})|\langle
z|X^{1/2}|\phi_{t}\rangle|^{2}\\
&  =\sum_{z,t}f((\nu_{t}/\mu_{z})^{-1})\lambda_{z}|\langle z|\phi_{t}%
\rangle|^{2}\\
&  \leq\sum_{z}f((\tau_{z}/\mu_{z})^{-1})\lambda_{z},
\end{align}
where $\tau_{z}=\sum_{t}|\langle z|\phi_{t}\rangle|^{2}\nu_{t}=\langle
z|\tau|z\rangle$. The inequality follows because the function $f(x^{-1})$ is
concave, due to the assumption that $f$ is operator anti-monotone with domain
$(0,\infty)$ and range $\mathbb{R}$.
\end{proof}

\bigskip

If $X$ and $Y$ have a classical--quantum form, as follows%
\begin{equation}
X=\sum_{z}|z\rangle\langle z|\otimes X^{z},\qquad Y=\sum_{z}|z\rangle\langle
z|\otimes Y^{z}, \label{eq:cq-form}%
\end{equation}
where $\{|z\rangle\}_{z}$ is an orthonormal basis and $\{X^{z}\}_{z}$ and
$\{Y^{z}\}_{z}$ are sets of positive semi-definite operators, then the
optimized $f$-divergence simplifies as well, generalizing
Proposition~\ref{prop:classical-case}.\ That is, it suffices to optimize over
positive definite states $\tau$ respecting the same classical--quantum form:

\begin{proposition}
[Classical--quantum case]Let $f$ be an operator anti-monotone function with
domain $(0,\infty)$ and range $\mathbb{R}$. Let $X$ and $Y$ be positive
semi-definite, having the classical--quantum form in \eqref{eq:cq-form}. Then%
\begin{equation}
\widetilde{Q}_{f}(X\Vert Y)=\sup_{\{\hat{\tau}^{z}\}_{z},\ \varepsilon>0}%
\sum_{z}\widetilde{Q}_{f}(X^{z}\Vert Y^{z}+\varepsilon\Pi_{Y_{z}}^{\perp}%
;\hat{\tau}^{z}),
\end{equation}
where each $\hat{\tau}^{z}$ is positive definite such that $\sum
_{z}\operatorname{Tr}\{\hat{\tau}^{z}\}=1$.
\end{proposition}

\begin{proof}
The main idea here is to show that the optimal $\tau$ takes on a
classical--quantum form as well, as $\tau=\sum_{z}|z\rangle\langle
z|\otimes\tau^{z}$. This follows from an application of the operator Jensen
inequality \cite{HP03}, as shown below. We focus on the case in which each
$Y^{z}$ is invertible. We adopt system labels $Z$ for the classical system and
$A$ for the quantum system. For a given positive definite $\tau$ with
$\operatorname{Tr}\{\tau\}=1$, we have that%
\begin{align}
\widetilde{Q}_{f}(X\Vert Y;\tau)  &  =\langle\varphi^{X}|_{ZA\hat{Z}\hat{A}%
}f(\tau_{ZA}^{-1}\otimes Y_{\hat{Z}\hat{A}}^{T})|\varphi^{X}\rangle_{ZA\hat
{Z}\hat{A}}\\
&  =\langle\varphi^{X}|_{ZA\hat{Z}\hat{A}}f\left(  \tau_{ZA}^{-1}\otimes
\sum_{z}|z\rangle\langle z|_{\hat{Z}}\otimes(Y_{\hat{A}}^{z})^{T}\right)
|\varphi^{X}\rangle_{ZA\hat{Z}\hat{A}}\\
&  =\langle\varphi^{X}|_{ZA\hat{Z}\hat{A}}\sum_{z}|z\rangle\langle z|_{\hat
{Z}}\otimes f\left(  \tau_{ZA}^{-1}\otimes(Y_{\hat{A}}^{z})^{T}\right)
|\varphi^{X}\rangle_{ZA\hat{Z}\hat{A}}.
\end{align}
Consider that $|z\rangle\langle z|_{\hat{Z}}$ is invariant under the action of
the decoherence or \textquotedblleft pinching\textquotedblright\ channel%
\begin{equation}
(\cdot)\rightarrow\mathcal{D}_{\hat{Z}}(\cdot)=\sum_{z}|z\rangle\langle
z|_{\hat{Z}}(\cdot)|z\rangle\langle z|_{\hat{Z}}.
\end{equation}
This implies that%
\begin{equation}
\langle\varphi^{X}|_{ZA\hat{Z}\hat{A}}f(\tau_{ZA}^{-1}\otimes Y_{\hat{Z}%
\hat{A}}^{T})|\varphi^{X}\rangle_{ZA\hat{Z}\hat{A}}=\langle\varphi
^{X}|_{ZA\hat{Z}\hat{A}}\mathcal{D}_{\hat{Z}}\left[  f(\tau_{ZA}^{-1}\otimes
Y_{\hat{Z}\hat{A}}^{T})\right]  |\varphi^{X}\rangle_{ZA\hat{Z}\hat{A}}.
\end{equation}
By \eqref{eq:transpose-trick}, the fact that $X_{ZA}=\mathcal{D}_{Z}(X_{ZA})$,
and defining $g(x)=f(x^{-1})$, we find that%
\begin{align}
\langle\varphi^{X}|_{ZA\hat{Z}\hat{A}}\mathcal{D}_{\hat{Z}}\left[  f(\tau
_{ZA}^{-1}\otimes Y_{\hat{Z}\hat{A}}^{T})\right]  |\varphi^{X}\rangle
_{ZA\hat{Z}\hat{A}}  &  =\langle\varphi^{X}|_{ZA\hat{Z}\hat{A}}\mathcal{D}%
_{Z}\left[  f(\tau_{ZA}^{-1}\otimes Y_{\hat{Z}\hat{A}}^{T})\right]
|\varphi^{X}\rangle_{ZA\hat{Z}\hat{A}}\\
&  =\langle\varphi^{X}|_{ZA\hat{Z}\hat{A}}\mathcal{D}_{Z}\left[  g(\tau
_{ZA}\otimes(Y_{\hat{Z}\hat{A}}^{T})^{-1})\right]  |\varphi^{X}\rangle
_{ZA\hat{Z}\hat{A}}\\
&  \leq\langle\varphi^{X}|_{ZA\hat{Z}\hat{A}}\left[  g(\mathcal{D}_{Z}%
(\tau_{ZA})\otimes(Y_{\hat{Z}\hat{A}}^{T})^{-1})\right]  |\varphi^{X}%
\rangle_{ZA\hat{Z}\hat{A}}\\
&  =\langle\varphi^{X}|_{ZA\hat{Z}\hat{A}}f(\left[  \mathcal{D}_{Z}(\tau
_{ZA})\right]  ^{-1}\otimes Y_{\hat{Z}\hat{A}}^{T})|\varphi^{X}\rangle
_{ZA\hat{Z}\hat{A}}.
\end{align}
The inequality follows from the operator Jensen inequality \cite{HP03} and the
fact that $g(x)$ is operator concave with domain $(0,\infty)$ and range
$\mathbb{R}$. Consider that%
\begin{equation}
\mathcal{D}_{Z}(\tau_{ZA})=\sum_{z}|z\rangle\langle z|_{Z}\otimes\hat{\tau
}_{A}^{z},
\end{equation}
for some $\{\tau^{z}\}_{z}$, where each $\hat{\tau}_{A}^{z}$ is positive
definite and $\sum_{z}\operatorname{Tr}\{\hat{\tau}_{A}^{z}\}=1$. Now consider
that%
\begin{align}
&  \langle\varphi^{X}|_{ZA\hat{Z}\hat{A}}f(\left[  \mathcal{D}_{Z}(\tau
_{ZA})\right]  ^{-1}\otimes Y_{\hat{Z}\hat{A}}^{T})|\varphi^{X}\rangle
_{ZA\hat{Z}\hat{A}}\nonumber\\
&  =\sum_{z,z^{\prime}}\langle\varphi^{X}|_{ZA\hat{Z}\hat{A}}\otimes
|z\rangle\langle z|_{Z}\otimes|z^{\prime}\rangle\langle z^{\prime}|_{\hat{Z}%
}\otimes f\left(  (\hat{\tau}_{A}^{z})^{-1}\otimes(Y_{\hat{A}}^{z^{\prime}%
})^{T}\right)  |\varphi^{X}\rangle_{ZA\hat{Z}\hat{A}}\\
&  =\sum_{z}\langle\varphi^{X^{z}}|_{A\hat{A}}f\left(  (\hat{\tau}_{A}%
^{z})^{-1}\otimes(Y_{\hat{A}}^{z})^{T}\right)  |\varphi^{X^{z}}\rangle
_{A\hat{A}}\\
&  =\sum_{z}\widetilde{Q}_{f}(X^{z}\Vert Y^{z};\hat{\tau}^{z}),
\end{align}
where the second-to-last line follows because%
\begin{align}
|\varphi^{X}\rangle_{ZA\hat{Z}\hat{A}}  &  =(X_{ZA}^{1/2}\otimes I_{\hat
{Z}\hat{A}})|\Gamma\rangle_{Z\hat{Z}}|\Gamma\rangle_{A\hat{A}}\\
&  =\sum_{z}|z\rangle\langle z|_{Z}\otimes(X_{A}^{z})^{1/2}|\Gamma
\rangle_{Z\hat{Z}}|\Gamma\rangle_{A\hat{A}}\\
&  =\sum_{z}|z\rangle_{Z}|z\rangle_{\hat{Z}}(X_{A}^{z})^{1/2}|\Gamma
\rangle_{A\hat{A}}\\
&  =\sum_{z}|z\rangle_{Z}|z\rangle_{\hat{Z}}|\varphi^{X^{z}}\rangle_{A\hat{A}%
}.
\end{align}
This completes the proof after optimizing over $\{\hat{\tau}^{z}\}_{z}$
satisfying $\sum_{z}\operatorname{Tr}\{\hat{\tau}^{z}\}=1$. We handle the case
of non-invertible $Y$ by taking a supremum over $\varepsilon>0$ at the end.
\end{proof}

\section{Optimized quantum $f$-information measures}

\label{sec:f-info}It is well known that the quantum relative entropy is a
parent quantity for many information measures used in quantum information
theory (see, e.g., \cite{D11} or \cite[Chapter~11]{W17}). As such, once one
has a base relative entropy or divergence to work with, there is now a
relatively standard recipe for generating other information measures, such as
entropy, conditional entropy, coherent information, mutual information,
entanglement measures, and more generally resource measures. This method has
been used in many works now
\cite{VP98,D09,S10,MH11,WWY13,MDSFT13,B13monotone,GW15,TWW14,WTB16,KW17}. Each
of the resulting quantities then satisfies a particular kind of quantum data
processing inequality, which follows as a consequence of the monotonocity of
the underlying relative entropy.

With the above in mind, we now mention some different information measures
that can be derived from the optimized $f$-divergence and we state the data
processing inequality that they satisfy. In what follows, $\rho_{AB}$ is a
density operator acting on a tensor-product Hilbert space $\mathcal{H}%
_{A}\otimes\mathcal{H}_{B}$, $\rho_{A}=\operatorname{Tr}_{B}\{\rho_{AB}\}$,
and $f$ is an operator anti-monotone function with domain $(0,\infty)$ and
range $\mathbb{R}$. Let $\mathcal{W}_{A\rightarrow A^{\prime}}$ denote a
subunital channel, satisfying $\mathcal{W}_{A\rightarrow A^{\prime}}%
(I_{A})\leq I_{A^{\prime}}$, and let $\mathcal{N}_{A\rightarrow A^{\prime}}$
and $\mathcal{M}_{B\rightarrow B^{\prime}}$ be quantum channels. All
statements about data processing follow from Theorem~\ref{thm:DP}\ and some
slight extra reasoning (see, e.g., \cite[Section~11.9]{W17}). One can find
various operational interpretations of entropic quantities discussed in
\cite{W17,H06book,H12}.

\begin{enumerate}
\item The optimized $f$-entropy is defined as%
\begin{equation}
\widetilde{S}_{f}(A)_{\rho}\equiv\widetilde{S}_{f}(\rho_{A})\equiv
-\widetilde{Q}_{f}(\rho_{A}\Vert I_{A}).
\end{equation}
It does not decrease under the action of a subunital channel $\mathcal{W}%
_{A\rightarrow A^{\prime}}$, in the sense that%
\begin{equation}
\widetilde{S}_{f}(A)_{\rho}\leq\widetilde{S}_{f}(A^{\prime})_{\mathcal{W}%
(\rho)}.
\end{equation}

\item The optimized $f$-mutual information is defined as%
\begin{equation}
\widetilde{I}_{f}(A;B)_{\rho}\equiv\inf_{\sigma_{B}}\widetilde{Q}_{f}%
(\rho_{AB}\Vert\rho_{A}\otimes\sigma_{B}).
\end{equation}
It does not increase under the action of the product channel $\mathcal{N}%
_{A\rightarrow A^{\prime}}\otimes\mathcal{M}_{B\rightarrow B^{\prime}}$, in
the sense that%
\begin{equation}
\widetilde{I}_{f}(A;B)_{\rho}\geq\widetilde{I}_{f}(A^{\prime};B^{\prime
})_{(\mathcal{N}\otimes\mathcal{M})(\rho)}.
\end{equation}

\item The optimized conditional $f$-entropy is defined as%
\begin{equation}
\widetilde{S}_{f}(A|B)_{\rho}\equiv-\inf_{\sigma_{B}}\widetilde{Q}_{f}%
(\rho_{AB}\Vert I_{A}\otimes\sigma_{B}).
\end{equation}
It does not decrease under the action of the product channel $\mathcal{W}%
_{A\rightarrow A^{\prime}}\otimes\mathcal{M}_{B\rightarrow B^{\prime}}$:%
\begin{equation}
\widetilde{S}_{f}(A|B)_{\rho}\leq\widetilde{S}_{f}(A|B)_{(\mathcal{W}%
\otimes\mathcal{M})(\rho)}.
\end{equation}

\item Related to the above, the optimized $f$-coherent information is defined
as%
\begin{equation}
\widetilde{I}_{f}(A\rangle B)_{\rho}\equiv-\widetilde{S}_{f}(A|B)_{\rho},
\end{equation}
and we have that%
\begin{equation}
\widetilde{I}_{f}(A\rangle B)_{\rho}\geq\widetilde{I}_{f}(A\rangle
B)_{(\mathcal{W}\otimes\mathcal{M})(\rho)}.
\end{equation}

\item In recent years, there has been much activity surrounding quantum
resource theories \cite{BG15,fritz_2015,RKR17,KR17,CG18}. Such a resource theory
consists of a few basic elements.\ There is a set $\mathcal{F}$\ of free
quantum states, i.e., those that the players involved are allowed to access
without any cost. Related to these, there is a set of free channels, and they
should have the property that a free state remains free after a free channel
acts on it. Once these are defined, it follows that any state that is not free
is considered resourceful, i.e., useful in the context of the resource theory.
We can also then define a measure of the resourcefulness of a quantum state,
and some fundamental properties that it should satisfy are that 1) it should
be monotone non-increasing under the action of a free channel and 2)\ it
should be equal to zero when evaluated on a free state. A typical choice of a
resourcefulness measure of a state $\rho$ satisfying these requirements is the
relative entropy of resourcefulness, defined in terms of relative entropy
as\ $\inf_{\sigma\in F}D(\rho\Vert\sigma)$. We can thus consider an optimized
$f$-relative entropy of resourcefulness as%
\begin{equation}
\widetilde{R}_{f}(\rho)\equiv\inf_{\sigma\in\mathcal{F}}\widetilde{Q}_{f}%
(\rho\Vert\sigma),
\end{equation}
and it thus satisfies the following data processing inequality%
\begin{equation}
\widetilde{R}_{f}(\rho)\geq\widetilde{R}_{f}(\mathcal{N}(\rho)),
\end{equation}
whenever $\mathcal{N}$ is a free channel as described above.

\item We can extend all of the above measures to quantum channel measures by
optimizing over inputs to the channel. For example, optimized $f$-mutual
information of a channel $\mathcal{N}_{A\rightarrow B}$ is defined as%
\begin{equation}
\sup_{\psi_{RA}}\widetilde{I}_{f}(R;B)_{\omega},
\end{equation}
where $\omega_{RB}=\mathcal{N}_{A\rightarrow B}(\psi_{RA})$ and $\psi_{RA}$ is
a pure bipartite state. Due to the Schmidt decomposition theorem and data
processing, it suffices to optimize over pure bipartite states $\psi_{RA}$
with the reference system $R$ isomorphic to the channel input system $A$.
\end{enumerate}

\subsection{Duality of optimized conditional $f$-entropy}

This paper's final contribution is the following proposition, which
generalizes a well known duality relation for conditional quantum entropy:

\begin{proposition}
[Duality]Let $f$ be an operator anti-monotone function with domain
$(0,\infty)$ and range $\mathbb{R}$. For a pure state $|\psi\rangle\langle
\psi|_{ABC}$, we have that%
\begin{equation}
\widetilde{S}_{f}(A|B)_{\psi}=-\widetilde{S}_{k}(A|C)_{\psi},
\end{equation}
where $k(x)=-f(x^{-1})$.
\end{proposition}

\begin{proof}
The method of proof is related to that given in \cite{B13monotone,MDSFT13}.
Set $\rho_{AB}=\operatorname{Tr}_{C}\{|\psi\rangle\langle\psi|_{ABC}\}$ and
consider that%
\begin{align}
\widetilde{S}_{f}(A|B)_{\psi}  &  =-\inf_{\sigma_{B}}\widetilde{Q}_{f}%
(\rho_{AB}\Vert I_{A}\otimes\sigma_{B})\\
&  =-\inf_{\sigma_{B}}\sup_{\tau_{AB}}\langle\varphi^{\rho_{AB}}|_{AB\hat
{A}\hat{B}}f(\tau_{AB}^{-1}\otimes I_{\hat{A}}\otimes\sigma_{\hat{B}}%
^{T})|\varphi^{\rho_{AB}}\rangle_{AB\hat{A}\hat{B}}\\
&  =-\sup_{\tau_{AB}}\inf_{\sigma_{B}}\langle\varphi^{\rho_{AB}}|_{AB\hat
{A}\hat{B}}f(\tau_{AB}^{-1}\otimes I_{\hat{A}}\otimes\sigma_{\hat{B}}%
^{T})|\varphi^{\rho_{AB}}\rangle_{AB\hat{A}\hat{B}}\\
&  =-\sup_{\tau_{AB}}\inf_{\sigma_{B}}\langle\varphi^{\rho_{AB}}|_{AB\hat
{A}\hat{B}}f(I_{A}\otimes\sigma_{B}\otimes\left(  \tau_{\hat{A}\hat{B}}%
^{-1}\right)  ^{T})|\varphi^{\rho_{AB}}\rangle_{AB\hat{A}\hat{B}}\\
&  =-\sup_{\tau_{C}}\inf_{\sigma_{B}}\langle\psi|_{ABC}f(I_{A}\otimes
\sigma_{B}\otimes\left(  \tau_{C}^{-1}\right)  ^{T})|\psi\rangle_{ABC}%
\end{align}
The first two equalities follow by definition. For simplicity, we consider
$\sigma_{B}$ to be an invertible density operator. The third equality follows
from an application of the minimax theorem \cite{S58}, considering that $f(x)$
is operator convex and $f(x^{-1})$ is operator concave. The fourth equality
follows by applying
\eqref{eq:f-divergence-spec-form}--\eqref{eq:gamma-to-trace}. The fifth
equality follows because $|\varphi^{\rho_{AB}}\rangle_{AB\hat{A}\hat{B}}$ and
$|\psi\rangle_{ABC}$ are purifications of $\rho_{AB}$ and all purifications
are related by an isometry (see, e.g., \cite{W17}). Furthermore, we have
isometric invariance of the optimized $\langle\varphi^{\rho_{AB}}|_{AB\hat
{A}\hat{B}}f(I_{A}\otimes\sigma_{B}\otimes( \tau_{\hat{A}\hat{B}}^{-1})
^{T})|\varphi^{\rho_{AB}}\rangle_{AB\hat{A}\hat{B}}$ by the same reasoning as
given in the proof of Proposition~\ref{prop:iso-inv}. Continuing,%
\begin{align}
&  =-\sup_{\tau_{C}}\inf_{\sigma_{B}}\left[  -\langle\psi|_{ABC}%
\ k(I_{A}\otimes\sigma_{B}^{-1}\otimes\tau_{C}^{T})|\psi\rangle_{ABC}\right]
\\
&  =\inf_{\tau_{C}}\sup_{\sigma_{B}}\langle\psi|_{ABC}\ k (I_{A}\otimes
\sigma_{B}^{-1}\otimes\tau_{C}^{T})|\psi\rangle_{ABC}\\
&  =\inf_{\tau_{C}}\widetilde{Q}_{k}(\psi_{AC}\Vert I_{A}\otimes\tau_{C})\\
&  =-\widetilde{S}_{k}(A|C)_{\psi}.
\end{align}
The first equality follows from the definition of the function $k$. The second
equality follows from propagating the inside minus sign to the outside. The
last equalities follow from applying similar steps as in the beginning of the
proof and then the definitions of $\widetilde{Q}_{k}(\rho_{AC}\Vert
I_{A}\otimes\tau_{C})$ and $\widetilde{S}_{k}(A|C)_{\psi}$.
\end{proof}

When $f(x)=x^{\left(  1-\alpha\right)  /\alpha}$ for $\alpha\in(1,\infty]$, we
recover a duality relation similar to that for sandwiched R\'enyi relative
entropy \cite{B13monotone,MDSFT13}. Duality relations for conditional entropy
and the data processing inequality are known to be closely related to entropic
uncertainty relations \cite{CCYZ11}, so there could be interesting new ones to
develop by choosing more general operator anti-monotone functions.

\section{Conclusion}

\label{sec:conclusion}The main contribution of the present paper is the
definition of the optimized quantum $f$-divergence and the proof that the data
processing inequality holds for it whenever the function $f$ is operator
anti-monotone with domain $(0,\infty)$ and range $\mathbb{R}$. The proof of
the data processing inequality relies on the operator Jensen inequality
\cite{HP03}, and it bears some similarities to the original approach from
\cite{P85,P86,TCR09}. Furthermore, I showed how the sandwiched R\'enyi
relative entropies are particular examples of the optimized quantum
$f$-divergence. As such, one benefit of this paper is that there is now a
single, unified approach, based on the operator Jensen inequality \cite{HP03},
for establishing the data processing inequality for the Petz--R\'enyi and
sandwiched R\'enyi relative entropies, for the full range of parameters for
which it is known to hold. In the remainder of the paper, I considered other
aspects such as the classical case, the classical--quantum case, and
information measures that one could construct from the optimized $f$-divergence.

There are several directions that one could pursue going forward. Equation
\eqref{eq:rel-mod-op} represents the function underlying the optimized
$f$-divergence in the relative modular operator formalism---this should be
helpful in understanding the optimized $f$-divergence in more general
contexts. Combined with the methods of \cite{P85,P86} and the approach in this
paper, it is clear that the data processing inequality will hold in more
general contexts. It would also be interesting to show that the data
processing inequality holds for maps beyond quantum channels, such as the
Schwarz and stochastic maps considered in \cite{HMPB11}. I suspect that the
methods of \cite{HMPB11} and the present paper could be used to establish the
data processing inequality for more general classes of maps.

\bigskip

\textbf{Acknowledgements.} I thank Milan Mosonyi and Anna Vershynina for
discussions related to the topic of this paper, and I acknowledge support from
the National Science Foundation under grant no.~1714215. I also thank the anonymous referees for many suggestions that helped to improve the paper.

\appendix

\section{Final case for Proposition~\ref{prop:mono-PT}}

\textbf{Case }$X_{A}$\textbf{ not invertible:}\ We now discuss how to extend
the proof detailed in the main text to the case in which $X_{A}$ is not
invertible (and thus $X_{AB}$ is not invertible either). In this case, with
$X_{A}^{-1/2}$ understood as a square-root inverse of $X_{A}$ on its support,
the Petz recovery map in \eqref{eq:petz-recovery}\ is no longer a quantum
channel, but it is instead a completely positive and trace non-increasing map.
A standard method for producing a quantum channel from the map in
\eqref{eq:petz-recovery}\ is to specify an additional action on the kernel of
$X_{A}$, as%
\begin{equation}
Z_{A}\rightarrow X_{AB}^{1/2}\left(  \left[  X_{A}^{-1/2}Z_{A}X_{A}%
^{-1/2}\right]  \otimes I_{B}\right)  X_{AB}^{1/2}+\operatorname{Tr}%
\{\Pi_{X_{A}}^{\perp}Z_{A}\}\xi_{AB}, \label{eq:extended-petz-channel}%
\end{equation}
where we take $\xi_{AB}$ to be an invertible density operator (see, e.g.,
\cite[Chapter~12]{W17} for this standard construction). One can check that the
map in \eqref{eq:extended-petz-channel} is completely positive and trace
preserving, and furthermore, an invertible input state leads to an invertible
output state. Our goal is now to find a Kraus decomposition for the above
quantum channel, so that we can work with its isometric extension as we did
previously. To begin with, suppose that the invertible state $\xi_{AB}$ has a
spectral decomposition as%
\begin{equation}
\xi_{AB}=\sum_{l=1}^{\left\vert A\right\vert \left\vert B\right\vert }%
p_{l}|\phi_{l}\rangle\langle\phi_{l}|_{AB},
\end{equation}
where $\{p_{l}\}_{l}$ is a probability distribution and $\{|\phi_{l}%
\rangle_{AB}\}_{l}$ is an orthonormal basis. Then we can write the channel in
\eqref{eq:extended-petz-channel}\ as%
\begin{multline}
\sum_{j=1}^{\left\vert B\right\vert }X_{AB}^{1/2}\left(  \left[  X_{A}%
^{-1/2}Z_{A}X_{A}^{-1/2}\right]  \otimes|j\rangle\langle j|_{B}\right)
X_{AB}^{1/2}+\sum_{k=1}^{\left\vert A\right\vert }\langle k|_{A}\Pi_{X_{A}%
}^{\perp}Z_{A}\Pi_{X_{A}}^{\perp}|k\rangle_{A}\xi_{AB}\\
=\sum_{j=1}^{\left\vert B\right\vert }X_{AB}^{1/2}\left[  X_{A}^{-1/2}%
\otimes|j\rangle_{B}\right]  Z_{A}\left[  X_{A}^{-1/2}\otimes\langle
j|_{B}\right]  X_{AB}^{1/2}\\
+\sum_{k=1}^{\left\vert A\right\vert }\sum_{l=1}^{\left\vert A\right\vert
\left\vert B\right\vert }\sqrt{p_{l}}|\phi_{l}\rangle_{AB}\langle k|_{A}%
\Pi_{X_{A}}^{\perp}Z_{A}\Pi_{X_{A}}^{\perp}|k\rangle_{A}\langle\phi_{l}%
|_{AB}\sqrt{p_{l}}.
\end{multline}
Thus, Kraus operators for it are as follows:%
\begin{equation}
\left\{  \left\{  X_{AB}^{1/2}\left[  X_{A}^{-1/2}\otimes|j\rangle_{B}\right]
\right\}  _{j=1}^{\left\vert B\right\vert },\left\{  \sqrt{p_{l}}|\phi
_{l}\rangle_{AB}\langle k|_{A}\Pi_{X_{A}}^{\perp}\right\}  _{l\in\left\{
1,\ldots,\left\vert A\right\vert \left\vert B\right\vert \right\}
,k\in\left\{  1,\ldots,\left\vert A\right\vert \right\}  }\right\}  .
\end{equation}
We now define an enlarged Hilbert space $\hat{C}$ to be the direct sum of
$\hat{B}$ and $\hat{A}$, and thus with dimension $|\hat{B}|+|\hat{A}|$, and an
orthonormal basis for it as%
\begin{equation}
\{|1\rangle_{\hat{C}},\ldots,|j\rangle_{\hat{C}},\ldots,||\hat{B}%
|\rangle_{\hat{C}},||\hat{B}|+1\rangle_{\hat{C}},\ldots,||\hat{B}%
|+k\rangle_{\hat{C}},\ldots,||\hat{B}|+|\hat{A}|\rangle_{\hat{C}}\}.
\end{equation}
We also define an auxiliary Hilbert space $E$ with orthonormal basis%
\begin{equation}
\{|e\rangle_{E},|1\rangle_{E},\ldots,||\hat{A}||\hat{B}|\rangle_{E}\},
\end{equation}
and we represent a purification $|\varphi^{\xi_{AB}}\rangle_{ABE}$ of the
state $\xi_{AB}$ as%
\begin{equation}
|\varphi^{\xi_{AB}}\rangle_{ABE}=\sum_{l=1}^{|\hat{A}||\hat{B}|}\sqrt{p_{l}%
}|\phi_{l}\rangle_{AB}|l\rangle_{E}.
\end{equation}
Thus, an isometric extension of the Petz recovery channel in
\eqref{eq:extended-petz-channel}, according to the standard recipe in
\eqref{eq:iso-from-kraus}, is given by%
\begin{multline}
\sum_{j=1}^{\left\vert B\right\vert }X_{AB}^{1/2}\left[  X_{A}^{-1/2}%
\otimes|j\rangle_{B}\right]  |j\rangle_{\hat{C}}|e\rangle_{E}+\sum
_{k=1}^{|\hat{A}|}\sum_{l=1}^{|\hat{A}||\hat{B}|}\sqrt{p_{l}}|\phi_{l}%
\rangle_{AB}\langle k|_{A}\Pi_{X_{A}}^{\perp}\otimes|k+\left\vert B\right\vert
\rangle_{\hat{C}}|l\rangle_{E}\\
=X_{AB}^{1/2}X_{A}^{-1/2}|\Gamma\rangle_{B\hat{C}}|e\rangle_{E}+|\varphi
^{\xi_{AB}}\rangle_{ABE}U_{A\rightarrow\hat{C}}\Pi_{X_{A}}^{\perp},
\end{multline}
where we set%
\begin{equation}
|\Gamma\rangle_{B\hat{C}}\equiv\sum_{j=1}^{\left\vert B\right\vert }%
|j\rangle_{B}|j\rangle_{\hat{C}},
\end{equation}
and we define the embedding map%
\begin{equation}
U_{A\rightarrow\hat{C}}\equiv\sum_{k=1}^{|\hat{A}|}|k+\left\vert B\right\vert
\rangle_{\hat{C}}\langle k|_{A}.
\end{equation}
So we set the isometry $V_{A\rightarrow B\hat{C}E}$ as%
\begin{equation}
V_{A\rightarrow B\hat{C}E}\equiv X_{AB}^{1/2}X_{A}^{-1/2}|\Gamma\rangle
_{B\hat{C}}|e\rangle_{E}+|\varphi^{\xi_{AB}}\rangle_{ABE}U_{A\rightarrow
\hat{C}}\Pi_{X_{A}}^{\perp}. \label{eq:iso-ext-2nd-case}%
\end{equation}
By construction, the operator $V_{A\rightarrow B\hat{C}E}$ is an isometry, but
we can verify by the following alternative calculation that this operator is
indeed an isometry:%
\begin{align}
V^{\dag}V  &  =\left(  \langle\Gamma|_{B\hat{C}}\langle e|_{E}X_{A}%
^{-1/2}X_{AB}^{1/2}+\Pi_{X_{A}}^{\perp}\left(  U_{A\rightarrow\hat{C}}\right)
^{\dag}\langle\varphi^{\xi_{AB}}|_{ABE}\right)  \times\nonumber\\
&  \qquad\left(  X_{AB}^{1/2}X_{A}^{-1/2}|\Gamma\rangle_{B\hat{C}}%
|e\rangle_{E}+|\varphi^{\xi_{AB}}\rangle_{ABE}U_{A\rightarrow\hat{C}}%
\Pi_{X_{A}}^{\perp}\right) \\
&  =\langle\Gamma|_{B\hat{C}}\langle e|_{E}X_{A}^{-1/2}X_{AB}^{1/2}%
X_{AB}^{1/2}X_{A}^{-1/2}|\Gamma\rangle_{B\hat{C}}|e\rangle_{E}\nonumber\\
&  \qquad+\Pi_{X_{A}}^{\perp}\left(  U_{A\rightarrow\hat{C}}\right)  ^{\dag
}\langle\varphi^{\xi_{AB}}|_{ABE}X_{AB}^{1/2}X_{A}^{-1/2}|\Gamma\rangle
_{B\hat{C}}|e\rangle_{E}\nonumber\\
&  \qquad+\langle\Gamma|_{B\hat{C}}\langle e|_{E}X_{A}^{-1/2}X_{AB}%
^{1/2}|\varphi^{\xi_{AB}}\rangle_{ABE}U_{A\rightarrow\hat{C}}\Pi_{X_{A}%
}^{\perp}\nonumber\\
&  \qquad+\Pi_{X_{A}}^{\perp}\left(  U_{A\rightarrow\hat{C}}\right)  ^{\dag
}\langle\varphi^{\xi_{AB}}|_{ABE}|\varphi^{\xi_{AB}}\rangle_{ABE}%
U_{A\rightarrow\hat{C}}\Pi_{X_{A}}^{\perp}\\
&  =\langle e|_{E}X_{A}^{-1/2}\langle\Gamma|_{B\hat{C}}X_{AB}|\Gamma
\rangle_{B\hat{C}}X_{A}^{-1/2}|e\rangle_{E}+\Pi_{X_{A}}^{\perp}\\
&  =X_{A}^{-1/2}X_{A}X_{A}^{-1/2}+\Pi_{X_{A}}^{\perp}\\
&  =\Pi_{X_{A}}+\Pi_{X_{A}}^{\perp}\\
&  =I_{A}.
\end{align}
In the above, we have used the fact that $\langle e|_{E}|\varphi^{\xi_{AB}%
}\rangle_{ABE}=0$. Now we extend $V$ to%
\begin{equation}
V_{A\hat{A}\rightarrow B\hat{C}E}\equiv V_{A\rightarrow B\hat{C}E}\otimes
I_{\hat{A}},
\end{equation}
and observe that%
\begin{equation}
V_{A\hat{A}\rightarrow B\hat{C}E}|\varphi^{X_{A}}\rangle_{A\hat{A}}%
=|\varphi^{X_{AB}}\rangle_{A\hat{A}B\hat{C}}|e\rangle_{E}=(X_{AB}^{1/2}\otimes
I_{\hat{A}\hat{C}})|\Gamma\rangle_{A\hat{A}}|\Gamma\rangle_{B\hat{C}}%
|e\rangle_{E}.
\end{equation}
Let $\tau_{AB}$ be the output of the Petz recovery channel when the invertible
state $\omega_{A}$ is input:%
\begin{equation}
\tau_{AB}=X_{AB}^{1/2}\left(  \left[  X_{A}^{-1/2}\omega_{A}X_{A}%
^{-1/2}\right]  \otimes I_{B}\right)  X_{AB}^{1/2}+\operatorname{Tr}%
\{\Pi_{X_{A}}^{\perp}\omega_{A}\}\xi_{AB}.
\end{equation}
Note that $\tau_{AB}$ is invertible because we chose $\xi_{AB}$ to be
invertible. Consider the positive definite operator $Y_{\hat{A}\hat{B}}$,
whose $\hat{B}$ system we embed into $\operatorname{span}\{|1\rangle_{\hat{C}%
},\ldots,|j\rangle_{\hat{C}},\ldots,||\hat{B}|\rangle_{\hat{C}}\}$ of system
$\hat{C}$, calling the embedded operator $Y_{\hat{A}\hat{C}}$. We then have
that $Y_{\hat{A}\hat{C}}^{T}U_{A\rightarrow\hat{C}}=0$, and so we find that%
\begin{align}
&  V^{\dag}\left(  \tau_{AB}^{-1}\otimes Y_{\hat{A}\hat{C}}^{T}\otimes
I_{E}\right)  V\nonumber\\
&  =\left[  \left(  \langle\Gamma|_{B\hat{C}}\langle e|_{E}X_{A}^{-1/2}%
X_{AB}^{1/2}+\Pi_{X_{A}}^{\perp}\left(  U_{A\rightarrow\hat{C}}\right)
^{\dag}\langle\varphi^{\xi_{AB}}|_{ABE}\right)  \otimes I_{\hat{A}}\right]
\times\nonumber\\
&  \qquad\left(  \tau_{AB}^{-1}\otimes Y_{\hat{A}\hat{C}}^{T}\otimes
I_{E}\right)  \left[  \left(  X_{AB}^{1/2}X_{A}^{-1/2}|\Gamma\rangle_{B\hat
{C}}|e\rangle_{E}+|\varphi^{\xi_{AB}}\rangle_{ABE}U_{A\rightarrow\hat{C}}%
\Pi_{X_{A}}^{\perp}\right)  \otimes I_{\hat{A}}\right] \\
&  =\langle\Gamma|_{B\hat{C}}\langle e|_{E}X_{A}^{-1/2}X_{AB}^{1/2}\left(
\tau_{AB}^{-1}\otimes Y_{\hat{A}\hat{C}}^{T}\otimes I_{E}\right)  X_{AB}%
^{1/2}X_{A}^{-1/2}|\Gamma\rangle_{B\hat{C}}|e\rangle_{E}\nonumber\\
&  \qquad+\Pi_{X_{A}}^{\perp}\left(  U_{A\rightarrow\hat{C}}\right)  ^{\dag
}\langle\varphi^{\xi_{AB}}|_{ABE}\left(  \tau_{AB}^{-1}\otimes Y_{\hat{A}%
\hat{C}}^{T}\otimes I_{E}\right)  X_{AB}^{1/2}X_{A}^{-1/2}|\Gamma
\rangle_{B\hat{C}}|e\rangle_{E}\nonumber\\
&  \qquad+\langle\Gamma|_{B\hat{C}}\langle e|_{E}X_{A}^{-1/2}X_{AB}%
^{1/2}\left(  \tau_{AB}^{-1}\otimes Y_{\hat{A}\hat{C}}^{T}\otimes
I_{E}\right)  |\varphi^{\xi_{AB}}\rangle_{ABE}U_{A\rightarrow\hat{C}}%
\Pi_{X_{A}}^{\perp}\nonumber\\
&  \qquad+\Pi_{X_{A}}^{\perp}\left(  U_{A\rightarrow\hat{C}}\right)  ^{\dag
}\langle\varphi^{\xi_{AB}}|_{ABE}\left(  \tau_{AB}^{-1}\otimes Y_{\hat{A}%
\hat{C}}^{T}\otimes I_{E}\right)  |\varphi^{\xi_{AB}}\rangle_{ABE}%
U_{A\rightarrow\hat{C}}\Pi_{X_{A}}^{\perp}.
\end{align}
The last three terms are equal to zero because $\langle\varphi^{\xi_{AB}%
}|_{ABE}|e\rangle_{E}=0$ and $Y_{\hat{A}\hat{C}}^{T}U_{A\rightarrow\hat{C}}%
=0$. Continuing, the last expression above is equal to%
\begin{align}
&  \langle\Gamma|_{B\hat{C}}\langle e|_{E}X_{A}^{-1/2}X_{AB}^{1/2}\left(
\tau_{AB}^{-1}\otimes Y_{\hat{A}\hat{C}}^{T}\otimes I_{E}\right)  X_{AB}%
^{1/2}X_{A}^{-1/2}|\Gamma\rangle_{B\hat{C}}|e\rangle_{E}\nonumber\\
&  =\langle\Gamma|_{B\hat{C}}\omega_{A}^{-1/2}\omega_{A}^{1/2}X_{A}%
^{-1/2}X_{AB}^{1/2}\left(  \tau_{AB}^{-1}\otimes Y_{\hat{A}\hat{C}}%
^{T}\right)  X_{AB}^{1/2}X_{A}^{-1/2}\omega_{A}^{1/2}\omega_{A}^{-1/2}%
|\Gamma\rangle_{B\hat{C}}\\
&  =\langle\Gamma|_{B\hat{C}}\left(  \omega_{A}^{-1/2}\left[  \omega_{A}%
^{1/2}X_{A}^{-1/2}X_{AB}^{1/2}\tau_{AB}^{-1}X_{AB}^{1/2}X_{A}^{-1/2}\omega
_{A}^{1/2}\right]  \omega_{A}^{-1/2}\otimes Y_{\hat{A}\hat{C}}^{T}\right)
|\Gamma\rangle_{B\hat{C}}\\
&  \leq\left\Vert \omega_{A}^{1/2}X_{A}^{-1/2}X_{AB}^{1/2}\tau_{AB}^{-1}%
X_{AB}^{1/2}X_{A}^{-1/2}\omega_{A}^{1/2}\right\Vert _{\infty}\langle
\Gamma|_{B\hat{C}}\left(  \omega_{A}^{-1/2}\omega_{A}^{-1/2}\otimes Y_{\hat
{A}\hat{C}}^{T}\right)  |\Gamma\rangle_{B\hat{C}}\\
&  =\left\Vert \tau_{AB}^{-1/2}X_{AB}^{1/2}X_{A}^{-1/2}\omega_{A}X_{A}%
^{-1/2}X_{AB}^{1/2}\tau_{AB}^{-1/2}\right\Vert _{\infty}\langle\Gamma
|_{B\hat{C}}\left(  \omega_{A}^{-1}\otimes Y_{\hat{A}\hat{C}}^{T}\right)
|\Gamma\rangle_{B\hat{C}}\\
&  \leq\langle\Gamma|_{B\hat{C}}\left(  \omega_{A}^{-1}\otimes Y_{\hat{A}%
\hat{C}}^{T}\right)  |\Gamma\rangle_{B\hat{C}}\\
&  =\omega_{A}^{-1}\otimes Y_{\hat{A}}^{T}.
\end{align}
The third equality follows because $\left\Vert Z^{\dag}Z\right\Vert _{\infty
}=\left\Vert ZZ^{\dag}\right\Vert _{\infty}$ for an operator $Z$. The last
inequality follows because for a positive semi-definite operator $W$, we have
that $\left\Vert W\right\Vert _{\infty}=\inf\left\{  \mu:W\leq\mu I\right\}
$. Applying this, we find that%
\begin{align}
&  \inf\left\{  \mu:\tau_{AB}^{-1/2}X_{AB}^{1/2}X_{A}^{-1/2}\omega_{A}%
X_{A}^{-1/2}X_{AB}^{1/2}\tau_{AB}^{-1/2}\leq\mu I_{AB}\right\} \nonumber\\
&  =\inf\left\{  \mu:X_{AB}^{1/2}X_{A}^{-1/2}\omega_{A}X_{A}^{-1/2}%
X_{AB}^{1/2}\leq\mu\ \tau_{AB}\right\} \\
&  =\inf\left\{  \mu:X_{AB}^{1/2}X_{A}^{-1/2}\omega_{A}X_{A}^{-1/2}%
X_{AB}^{1/2}\leq\mu\left[  X_{AB}^{1/2}X_{A}^{-1/2}\omega_{A}X_{A}%
^{-1/2}X_{AB}^{1/2}+\operatorname{Tr}\{\Pi_{X_{A}}^{\perp}\omega_{A}\}\xi
_{AB}\right]  \right\} \\
&  \leq1.
\end{align}

Let $P_{\hat{C}}$ be the embedding of the identity operator $I_{A}$ into the
subspace of $\hat{C}$ spanned by%
\begin{equation}
\{||\hat{B}|+1\rangle_{\hat{C}},\ldots,||\hat{B}|+k\rangle_{\hat{C}}%
,\ldots,||\hat{B}|+|\hat{A}|\rangle_{\hat{C}}\}.
\end{equation}
That is, $P_{\hat{C}}\equiv\sum_{k=1}^{|\hat{A}|}||\hat{B}|+k\rangle_{\hat{C}%
}\langle|\hat{B}|+k|_{\hat{C}}. $ Consider that, due to $P_{\hat{C}}%
|\Gamma\rangle_{B\hat{C}}=0$, we have that%
\begin{align}
&  V^{\dag}\left(  \tau_{AB}^{-1}\otimes I_{\hat{A}}\otimes P_{\hat{C}}\otimes
I_{E}\right)  V\nonumber\\
&  =\left[  \left(  \langle\Gamma|_{B\hat{C}}\langle e|_{E}X_{A}^{-1/2}%
X_{AB}^{1/2}+\Pi_{X_{A}}^{\perp}\left(  U_{A\rightarrow\hat{C}}\right)
^{\dag}\langle\varphi^{\xi_{AB}}|_{ABE}\right)  \otimes I_{\hat{A}}\right]
\times\nonumber\\
&  \qquad\left(  \tau_{AB}^{-1}\otimes I_{\hat{A}}\otimes P_{\hat{C}}\otimes
I_{E}\right)  \left[  \left(  X_{AB}^{1/2}X_{A}^{-1/2}|\Gamma\rangle_{B\hat
{C}}|e\rangle_{E}+|\varphi^{\xi_{AB}}\rangle_{ABE}U_{A\rightarrow\hat{C}}%
\Pi_{X_{A}}^{\perp}\right)  \otimes I_{\hat{A}}\right] \\
&  =\langle\Gamma|_{B\hat{C}}\langle e|_{E}X_{A}^{-1/2}X_{AB}^{1/2}\left(
\tau_{AB}^{-1}\otimes I_{\hat{A}}\otimes P_{\hat{C}}\otimes I_{E}\right)
X_{AB}^{1/2}X_{A}^{-1/2}|\Gamma\rangle_{B\hat{C}}|e\rangle_{E}\nonumber\\
&  \qquad+\Pi_{X_{A}}^{\perp}\left(  U_{A\rightarrow\hat{C}}\right)  ^{\dag
}\langle\varphi^{\xi_{AB}}|_{ABE}\left(  \tau_{AB}^{-1}\otimes I_{\hat{A}%
}\otimes P_{\hat{C}}\otimes I_{E}\right)  X_{AB}^{1/2}X_{A}^{-1/2}%
|\Gamma\rangle_{B\hat{C}}|e\rangle_{E}\nonumber\\
&  \qquad+\langle\Gamma|_{B\hat{C}}\langle e|_{E}X_{A}^{-1/2}X_{AB}%
^{1/2}\left(  \tau_{AB}^{-1}\otimes I_{\hat{A}}\otimes P_{\hat{C}}\otimes
I_{E}\right)  |\varphi^{\xi_{AB}}\rangle_{ABE}U_{A\rightarrow\hat{C}}%
\Pi_{X_{A}}^{\perp}\nonumber\\
&  \qquad+\Pi_{X_{A}}^{\perp}\left(  U_{A\rightarrow\hat{C}}\right)  ^{\dag
}\langle\varphi^{\xi_{AB}}|_{ABE}\left(  \tau_{AB}^{-1}\otimes I_{\hat{A}%
}\otimes P_{\hat{C}}\otimes I_{E}\right)  |\varphi^{\xi_{AB}}\rangle
_{ABE}U_{A\rightarrow\hat{C}}\Pi_{X_{A}}^{\perp}\\
&  =\Pi_{X_{A}}^{\perp}\left(  U_{A\rightarrow\hat{C}}\right)  ^{\dag}%
\langle\varphi^{\xi_{AB}}|_{ABE}\left(  \tau_{AB}^{-1}\otimes I_{\hat{A}%
}\otimes P_{\hat{C}}\otimes I_{E}\right)  |\varphi^{\xi_{AB}}\rangle
_{ABE}U_{A\rightarrow\hat{C}}\Pi_{X_{A}}^{\perp}\\
&  =\Pi_{X_{A}}^{\perp}\otimes I_{\hat{A}}\ \langle\varphi^{\xi_{AB}}%
|_{ABE}\tau_{AB}^{-1}|\varphi^{\xi_{AB}}\rangle_{ABE}\\
&  =\Pi_{X_{A}}^{\perp}\otimes I_{\hat{A}}\ \operatorname{Tr}\{\xi_{AB}%
\tau_{AB}^{-1}\}.
\end{align}
Observe that%
\begin{align}
f\left(  \tau_{AB}^{-1}\otimes Y_{\hat{A}\hat{B}}^{T}\right)   &  =\langle
e|_{E}\left[  f\left(  \tau_{AB}^{-1}\otimes Y_{\hat{A}\hat{B}}^{T}\right)
\otimes I_{E}\right]  |e\rangle_{E}\\
&  =\langle e|_{E}\left[  f\left(  \tau_{AB}^{-1}\otimes Y_{\hat{A}\hat{B}%
}^{T}\otimes I_{E}\right)  \right]  |e\rangle_{E}.
\end{align}
Furthermore, consider that for an $\varepsilon\in(0,1)$, we have that%
\begin{align}
&  \widetilde{Q}_{f}(X_{AB}\Vert Y_{AB};\tau_{AB})\nonumber\\
&  =\langle\varphi^{X_{AB}}|_{A\hat{A}B\hat{B}}f\left(  \tau_{AB}^{-1}\otimes
Y_{\hat{A}\hat{B}}^{T}\right)  |\varphi^{X_{AB}}\rangle_{A\hat{A}B\hat{B}}\\
&  =\langle\varphi^{X_{AB}}|_{A\hat{A}B\hat{B}}\langle e|_{E}\left[  f\left(
\tau_{AB}^{-1}\otimes Y_{\hat{A}\hat{B}}^{T}\otimes I_{E}\right)  \right]
|\varphi^{X_{AB}}\rangle_{A\hat{A}B\hat{B}}|e\rangle_{E}\\
&  =\langle\varphi^{X_{AB}}|_{A\hat{A}B\hat{C}}\langle e|_{E}\left[  f\left(
\tau_{AB}^{-1}\otimes\left(  Y_{\hat{A}\hat{C}}^{T}+\varepsilon\left[
I_{\hat{A}}\otimes P_{\hat{C}}\right]  \right)  \otimes I_{E}\right)  \right]
|\varphi^{X_{AB}}\rangle_{A\hat{A}B\hat{C}}|e\rangle_{E}\\
&  =\langle\varphi^{X_{A}}|_{A\hat{A}}V^{\dag}\left[  f\left(  \tau_{AB}%
^{-1}\otimes\left(  Y_{\hat{A}\hat{C}}^{T}+\varepsilon\left[  I_{\hat{A}%
}\otimes P_{\hat{C}}\right]  \right)  \otimes I_{E}\right)  \right]
V|\varphi^{X_{A}}\rangle_{A\hat{A}}\\
&  \geq\langle\varphi^{X_{A}}|_{A\hat{A}}f\left(  V^{\dag}\left[  \tau
_{AB}^{-1}\otimes\left(  Y_{\hat{A}\hat{C}}^{T}+\varepsilon\left[  I_{\hat{A}%
}\otimes P_{\hat{C}}\right]  \right)  \otimes I_{E}\right]  V\right)
|\varphi^{X_{A}}\rangle_{A\hat{A}}\\
&  \geq\langle\varphi^{X_{A}}|_{A\hat{A}}f\left(  \omega_{A}^{-1}\otimes
Y_{\hat{A}}^{T}+\varepsilon\left[  \Pi_{X_{A}}^{\perp}\otimes I_{\hat{A}%
}\ \operatorname{Tr}\{\xi_{AB}\tau_{AB}^{-1}\}\right]  \right)  |\varphi
^{X_{A}}\rangle_{A\hat{A}}%
\end{align}
The third equality follows because the term $\varepsilon\left[  I_{\hat{A}%
}\otimes P_{\hat{C}}\right]  $ gets zeroed out\ due to the sandwich by
$|\varphi^{X_{AB}}\rangle_{A\hat{A}B\hat{C}}$, given that $|\varphi^{X_{AB}%
}\rangle_{A\hat{A}B\hat{C}}$ only has support in $\operatorname{span}%
\{|1\rangle_{\hat{C}},\ldots,|j\rangle_{\hat{C}},\ldots,||\hat{B}%
|\rangle_{\hat{C}}\}$ (this can be seen explicitly by examining the proof of
Proposition~\ref{prop:iso-inv}). Furthermore, note that the operator
$Y_{\hat{A}\hat{C}}^{T}+\varepsilon\left[  I_{\hat{A}}\otimes P_{\hat{C}%
}\right]  $ is invertible. The first inequality follows from the operator
Jensen inequality \cite{HP03}. The next inequality follows because%
\begin{align}
&  V^{\dag}\left[  \tau_{AB}^{-1}\otimes\left(  Y_{\hat{A}\hat{C}}%
^{T}+\varepsilon\left[  I_{\hat{A}}\otimes P_{\hat{C}}\right]  \right)
\otimes I_{E}\right]  V\nonumber\\
&  =V^{\dag}\left[  \tau_{AB}^{-1}\otimes Y_{\hat{A}\hat{C}}^{T}\otimes
I_{E}\right]  V+V^{\dag}\left[  \tau_{AB}^{-1}\otimes\varepsilon\left[
I_{\hat{A}}\otimes P_{\hat{C}}\right]  \otimes I_{E}\right]  V\\
&  \leq\omega_{A}^{-1}\otimes Y_{\hat{A}}^{T}+\varepsilon\Pi_{X_{A}}^{\perp
}\otimes I_{\hat{A}}\ \operatorname{Tr}\{\xi_{AB}\tau_{AB}^{-1}\},
\end{align}
and by applying operator anti-monotonicity of $f$. This establishes the
inequality for all $\varepsilon\in(0,1)$. Thus, we can apply continuity of $f$
and take the limit $\varepsilon\searrow0$ to find that%
\begin{equation}
\widetilde{Q}_{f}(X_{AB}\Vert Y_{AB};\tau_{AB})\geq\langle\varphi^{X_{A}%
}|_{A\hat{A}}f\left(  \omega_{A}^{-1}\otimes Y_{\hat{A}}^{T}\right)
|\varphi^{X_{A}}\rangle_{A\hat{A}}=\widetilde{Q}_{f}(X_{A}\Vert Y_{A}%
;\omega_{A}).
\end{equation}
We can now take the supremum over all invertible states $\tau_{AB}$ to get the
following inequality holding for all invertible states $\omega_{A}$:%
\begin{equation}
\widetilde{Q}_{f}(X_{AB}\Vert Y_{AB})\geq\widetilde{Q}_{f}(X_{A}\Vert
Y_{A};\omega_{A}).
\end{equation}
After taking a supremum over invertible states $\omega_{A}$, we find that the
inequality in \eqref{eq:DP-partial-trace} holds when $X_{A}$ is not invertible.

\begin{remark}
\label{rem:cont-appeal} Several of the works \cite{P86,PS09,P10a,S10} on
quantum $f$-divergence consider only invertible density operators and then
appeal to continuity in order to extend proofs to the whole set of density
operators. This is often understood as simply adding $\varepsilon I$ to a
density operator and then taking the limit $\varepsilon\rightarrow0$ later. In
the second case given in the proof of Theorem~\ref{thm:DP}, in which $X_{AB}$
is not invertible but $X_{A}$ is, the method can be understood as falling
under an appeal to continuity. However, in the last case detailed above, when $X_{A}$ is
not invertible, the method arguably goes beyond a mere appeal to continuity,
given the construction of the channel in \eqref{eq:extended-petz-channel}, the
corresponding isometric extension in \eqref{eq:iso-ext-2nd-case}, and the
ensuing analysis.
\end{remark}

\bibliographystyle{alpha}
\bibliography{Ref}

\end{document}